\documentclass[11pt,a4paper]{amsart}
\usepackage{a4wide}
\usepackage{amssymb}
\usepackage{amsmath,amsthm}
\usepackage{hyperref}
\usepackage{graphicx}
\usepackage{color}
\usepackage{float}

\newtheorem{theorem}{Theorem}[section]
\newtheorem{lemma}[theorem]{Lemma}
\newtheorem{proposition}[theorem]{Proposition}

\theoremstyle{definition}

\newcommand{\rom}[1]{\text{\sf #1}}

\def\1{\mathbb 1}

\renewcommand{\d}{{\mathrm{d}}}

\newcommand{\R}{\mathbb{R}}
\newcommand{\N}{\mathbb{N}}
\newcommand{\Up}{U_+}

\newcommand{\dd}{{\mathrm{d}}}

\newcommand{\Con}{\ensuremath{\mathcal{C}}}

\newcommand{\diam}{{\mathrm{diam}}}
\newcommand{\lip}{{\mathrm{Lip}}}
\newcommand{\ep}{\varepsilon}

\numberwithin{equation}{section}

\newcommand{\supp}{\mathrm{supp}}

\theoremstyle{remark}
\newtheorem{remark}[theorem]{Remark}

\newcommand{\veps}{\varepsilon}

\theoremstyle{remark}
\newcommand{\eps}{\veps}

\title{Geodesics in nonexpanding impulsive gravitational
waves with $\Lambda$, Part I}

\author[C. S\"amann]{Clemens S\"amann}
\address{Faculty of Mathematics, University of Vienna, Oskar-Morgenstern-Platz
1, 1090 Vienna, Austria}
\email{clemens.saemann@univie.ac.at}
\author[R. Steinbauer]{Roland Steinbauer}
\address{Faculty of Mathematics, University of Vienna, Oskar-Morgenstern-Platz
1, 1090 Vienna, Austria}
\email{roland.steinbauer@univie.ac.at}
\author[A. Lecke]{Alexander Lecke}
\address{Faculty of Mathematics, University of Vienna, Oskar-Morgenstern-Platz
1, 1090 Vienna, Austria}
\email{alexander.lecke@univie.ac.at}
\author[J. Podolsk\'y]{Ji\v{r}\'i Podolsk\'y}
\address{Institute of Theoretical Physics, Faculty of Mathematics and Physics,
Charles University in Prague, V~Hole\v{s}ovi\v{c}k\'ach~2, 180~00 Praha 8, Czech
Republic }
\email{podolsky@mbox.troja.mff.cuni.cz}
\date{\today}
\begin{document}

\begin{abstract}
We investigate the geodesics in the entire class of nonexpanding impulsive
gravitational waves propagating in an (anti-)de Sitter universe using the
distributional form of the metric. Employing a $5$-dimensional embedding
formalism and a general regularisation technique we prove existence and uniqueness
of geodesics crossing the wave impulse leading to a completeness result.
We also derive the explicit form of the geodesics thereby confirming previous results
derived in a heuristic approach.

\vskip 1em

\noindent
\emph{Keywords:} impulsive gravitational waves, geodesics, low regularity
\medskip

\noindent 
\emph{MSC2010:} 83C15, 
		83C35, 
		46F10, 
		34A36  

\noindent
\emph{PACS numbers:} 02.30.Hq, 
		     04.20.Jb, 
		     04.30.Nk  

\end{abstract}

\maketitle

\section{Introduction}

Impulsive \emph{pp}-waves have now been studied for several decades and have
become textbook examples of exact radiative spacetimes modelling short but
intense bursts of gravitational radiation propagating in a Minkowski background
(see e.g.\ \cite[Sec.\ 20.2]{GP:09} and the references therein).
Such geometries have been introduced by Roger Penrose using his
`scissors and paste method' (see e.g.\ \cite{Pen72}) leading to the
\emph{distributional Brinkmann form} of the metric
\begin{equation}\label{ipp}
 \d s^2=\d x^2+\d y^2-2\d u \d v+f(x,y)\delta(u)\d u^2,
\end{equation}
i.e., impulsive limits of sandwich \emph{pp}-waves \cite{BPR59}.
Alternatively, almost at the same time, Aichelburg and Sexl in \cite{AS} have obtained a special impulsive
\emph{pp}-wave as the ultrarelativistic limit of the Schwarzschild geometry and
several authors have applied the same approach to other solutions of the
Kerr-Newman family (see e.g.\ \cite[Ch.\ 4]{BarHog2003} and \cite[Sec.\ 3.5.1]{Pod2002b}
for an overview). This procedure of boosting static sources to the speed of light was later generalised
to the case of a non-vanishing cosmological constant $\Lambda$ in the pioneering work \cite{HotTan93}
by Hotta and Tanaka (see also \cite{[B5],[B6]}) which lead to an increased interest in nonexpanding impulsive
waves in cosmological de Sitter and anti-de Sitter backgrounds. The Penrose `scissors and paste method'
for non-vanishing $\Lambda$ was described in \cite{Sfet,[B7]} while impulsive limits in the Kundt class
were considered in \cite{[B1]} and elsewhere, see \cite[Sec.\ 20.3]{GP:09} for an overview.
\medskip

Generally, nonexpanding impulsive waves in all backgrounds of constant curvature
can be described by a continuous as well as by a  distributional form of
the metric tensor. To give a brief discussion of these we start with the
conformally flat form of Minkowski ($\Lambda=0$) and (anti-)de Sitter ($\Lambda\not=0$)
background spacetimes
\begin{equation}
\d s_0^2= \frac{2\,\d\eta\,\d\bar\eta
-2\,\d{\mathcal U}\,\d{\mathcal V}}{[\,1+{\frac{1}{6}}\Lambda(\eta\bar\eta-{\mathcal U}{\mathcal
 V})\,]^2}\,.
 \label{conf*}
\end{equation}
As in \cite{[B7]}, here ${\mathcal U}, \mathcal{V}$ are the usual null and $\eta,\bar\eta$  the
usual complex spatial coordinates. Now for ${\mathcal U}>0$ we perform the transformation
\begin{equation}\label{ro:trsf}
 {\mathcal U}=U\,,\quad {\mathcal V}=V+H+UH_{,Z}H_{,\bar Z}\,,\quad
\eta=Z+UH_{,\bar Z}\,,
\end{equation}
where $H(Z,\bar Z)$ is an arbitrary real-valued function.
Combining this with the background line element \eqref{conf*} for ${U<0}$ in
which we set ${{\mathcal U}=U}$, ${{\mathcal V}=V}$, ${\eta=Z}$  we obtain the \emph{continuous}
form of the metric
\begin{equation}\label{conti}
\d s^2= \frac{2\,|\d Z+\Up(H_{,Z\bar Z}\d Z+{H}_{,\bar Z\bar Z}\d\bar Z)|^2-2\,\d
U\d V}
{[\,1+\frac{1}{6}\Lambda(Z\bar Z-UV+\Up G)\,]^2}\,,
 \end{equation}
where\footnote{This choice of sign of $G$ is in accordance with
\cite{PO:01}, which is our main point of reference, but different from more
recent papers, e.g.\ \cite{PSSS:15}.} ${G(Z,\bar Z)= ZH_{,Z}+\bar ZH_{,\bar
Z}-H}$ and $\Up\equiv\Up(U)=0$
for $U\leq 0$ and $\Up(U)=U$ for $U\geq0$. This \emph{`kink-function'} $\Up$ is Lipschitz continuous, hence
the spacetime (apart from possible poles of
$H$, which indeed occur in physically realistic models, see e.g.\
\cite{AS,HotTan93}, and Section \ref{sec2}, below)
is locally Lipschitz. Recall that a locally Lipschitz metric
possesses a locally bounded connection and hence a distributional curvature,
which, however, in general is unbounded. In fact the
discontinuity in the derivatives of the metric introduces impulsive components
in the Weyl and curvature tensors (see \cite{[B7]}), and the metric
(\ref{conti}) explicitly describes  impulsive waves in de~Sitter, anti-de~Sitter
or Minkowski backgrounds.

For ${\Lambda=0}$, \eqref{conti} reduces to the classic Rosen form of
impulsive {\it pp\,}-waves (cf.\ \cite[Sec.\ 17.5]{GP:09}).
In this special case the geodesic equation has been rigorously solved in
\cite{LSS:13} using Carath\'eodory's solution concept (see e.g.\
\cite[Ch.\ 1]{F:88}), which allows to deal with the locally bounded but
discontinuous right hand side of the equation. The geodesics thereby obtained
coincide with the limits of the  geodesics for the distributional form
\eqref{ipp} calculated in \cite{KS:99} which have been previously derived
heuristically
(e.g.\ in \cite{DT,Bal1997,[C2]}).

To deal, however, with the geodesic equation for the full class of nonexpanding
impulsive waves for arbitrary $\Lambda$, that is the complete metric \eqref{conti},
the more sophisticated solution concept
of Filippov (\cite[Ch.\ 2]{F:88}) has been applied recently in \cite{PSSS:15}.
Building on a general result for all locally Lipschitz spacetimes (\cite{S:14}),
existence, uniqueness and global ${\mathcal C}^1$-regularity of the geodesics has been
established. This, in particular, justifies the ${\mathcal C}^1$-matching procedure
which has been used before to explicitly derive the geodesics
in this and similar situations (\cite{[C2],bis-proc,[B7],PO:01,PS03,PS10}).
\medskip

On the other hand, the \emph{distributional form} of the impulsive metric arises
by writing the transformation relating (\ref{conf*}) and (\ref{conti}) in a
combined way for all $U$ using the Heaviside function $\Theta$ as
\begin{equation} \label{trans}
 {\mathcal U}=U\,,\quad
 {\mathcal V}=V+\Theta\,H+\Up\,H_{,Z}H_{,\bar Z}\,, \quad
 \eta=Z+\Up\,H_{,\bar Z}\,.
\end{equation}
Of course, \eqref{trans} is discontinuous in the coordinate ${\mathcal V}$ and
merely Lipschitz continuous in $\eta$ across ${\{{\mathcal U}=0\}}$  but
applying it \emph{formally} to \eqref{conti} we
arrive at
 \begin{equation}
\d s^2= \frac{2\,\d\eta\,\d\bar\eta-2\,\d {\mathcal U}\,\d {\mathcal V}
+2H(\eta,\bar\eta)\,\delta({\mathcal U})\,\d {\mathcal U}^2}
{[\,1+\frac{1}{6}\Lambda(\eta\bar\eta-{\mathcal U}{\mathcal V})\,]^2}\,,
 \label{confppimp}
 \end{equation}
which has the striking advantage of coinciding with the background metric
$\d s_0^2$ \eqref{conf*}  off the impulse located at ${\mathcal U}=0$.
This, however, comes at the price of introducing
a distributional coefficient in the metric which leads us
out of the Geroch--Traschen class (\cite{GT:87}) of metrics (of regularity
$W^{1,2}_{\mbox{\small loc}}\cap L^\infty_{\mbox{\small loc}}$), which
guarantees existence and stability of the curvature in distributions (see also
\cite{LFM:07,SV:09}). However, due to its simple
geometrical structure the metric (\ref{confppimp})  nevertheless allows to
calculate the curvature as a distribution, again leading to impulsive
components in the Weyl and curvature tensors (\cite{[B7]}). Also, in
Minkowski background the metric \eqref{confppimp} reduces to the Brinkmann form of
impulsive {\it pp\,}-waves (\ref{ipp}).
\medskip

Clearly, a mathematically sound treatment of the transformation \eqref{trans} is a
delicate matter. In the special case of \emph{pp}-waves
this has been achieved in \cite{KS:99a} using nonlinear distributional geometry
(\cite[Ch.\ 3]{GKOS:01}) based on algebras of generalised functions (\cite{C:85}).
More precisely, the `discontinuous coordinate change' was shown to be the
distributional limit of a `generalised diffeomorphism', a concept further studied in
\cite{EG:11,EG:13}. The key to these results was provided by a nonlinear distributional analysis
of the geodesics in the metric \eqref{ipp}, and, in particular, an existence, uniqueness
and completeness result for the geodesic equation in nonlinear generalised functions (\cite{Steina,KS:99}).

This suggests that a first step towards the long-term goal of understanding the
transformation \eqref{trans} for $\Lambda\not=0$ is to reach a mathematically
sound understanding of the geodesics in the distributional metric
\eqref{confppimp}. The geodesic equation for \eqref{confppimp}, however,
displays a very  singular behaviour including terms proportional to the square
of the Dirac-$\delta$ (cf.\ \cite{Sfet}). For this reason the authors of
\cite{PO:01} have employed a five-dimensional formalism (\cite{[B6],[B7]}) of
embedding (anti-)de~Sitter space into a $5$-dimensional \emph{pp}-wave spacetime
(see Section \ref{sec2} below). In this approach the geodesic equation takes a
form that is distributionally accessible at all, however, not mathematically
rigorously. In the absence of a valid solution concept for this nonlinear
distributional equations a natural ansatz was used to derive the geodesics and
to study them in detail in \cite[Sec.\ 4--5]{PO:01}. Nevertheless, a desirable nonlinear
distributional analysis of the geodesic equation in the ${\Lambda\not=0}$-cases,
which will eventually lead to a mathematical understanding of the transformation
\eqref{trans}, has been missing to date.

In this paper we provide such an analysis. Thereby, we follow \cite{PO:01} in
employing the five-dimensional formalism. We will, however, not use any
theory of nonlinear distributions leaving a detailed study of nonexpanding
impulsive waves in (anti-)de Sitter universe as (nonlinear) distributional geometries to a subsequent paper.
Instead, we will employ a regularisation approach and view \eqref{confppimp} as a spacetime with a short but finitely
extended impulse (i.e., a generic sandwich wave with support in a regularisation strip which we will also
call the `wave zone') and employ
an analysis in the spirit of \cite{SS12} where impulsive limits in a class of \emph{pp}-wave type
spacetimes with a more general wave surface but vanishing $\Lambda$
(\cite{CFS03,FS03,CFS04,FS06}) have been considered.
\medskip

We will detail this regularisation approach in the next section after
introducing the $5$-dimen\-sio\-nal formalism.In particular, we will replace the Dirac-$\delta$ in the metric
\eqref{classical} below by a very general regularisation $\delta_\eps$, thereby 
ensuring that our results are \emph{regularisation independent} (within the
class of so called a \emph{model delta nets}).
Then in Section \ref{sec3} we will
employ a fixed point argument to show that the regularised equations have unique smooth solutions
which cross the regularisation strip. This will lead to our main result on
completeness of nonexpanding impulsive gravitational
waves in a cosmological background.  The technical proofs allowing for the application of the fixed
point theorem are deferred to Appendix \ref{A}. In Section \ref{sec4} we study boundedness
properties of the regularised geodesics which are essential when dealing with their limits
in Section \ref{sec5}. There we show that the solutions of the regularised geodesic equation
converge, as the regularisation parameter goes to zero, to geodesics of the background (anti-)de Sitter
spacetime which have to be matched ap\-pro\-priately across the impulse and have been derived previously
in \cite{PO:01}. Overly technical calculations are collected in Appendix \ref{B}.

\section{The geodesic equation for $\Lambda\not=0$}\label{sec2}

In this section we detail our regularisation approach and derive the respective geodesic equations.
To begin with, however, we recall the $5$-dimensional formalism of
\cite{[B6],[B7]} for the full
class of nonexpanding impulsive waves with non-vanishing $\Lambda$. To this end
we start with the 5-dimensional impulsive \textit{pp}-wave spacetime
$\overline M$ with metric (extending \eqref{ipp})
\begin{equation}\label{classical}
\d s^{2}=\d Z_{2}^{2}+\d Z_{3}^{2}+ \sigma
\d Z_{4}^{2}-2\d U \d V +H(Z_{2},Z_{3},Z_{4})\delta(U)\d U^{2}
\end{equation}
and consider the four-dimensional (anti-)de Sitter hyperboloid $(M,g)$ given by the constraint
\begin{equation}\label{const}
Z_{2}^{2}+Z_{3}^{2}+ \sigma Z_{4}^{2}-2UV=\sigma a^{2},
\end{equation}
where $a:=\sqrt{3/(\sigma \Lambda)},\ \sigma:=\text{sign}(\Lambda)=\pm 1$, and
${U=\frac{1}{\sqrt2}(Z_0+Z_1)}$, ${V=\frac{1}{\sqrt2}(Z_0-Z_1)}$ are null-coordinates.\footnote{These coordinates are different from those used in the metric \eqref{conti}. Since in this paper we will not use the continuous form \eqref{conti} we simplify the notation by not distinguishing them by a bar (as we did in \cite{PSSS:15}).} Here
$(Z_0,\dots,Z_4)$ are global Cartesian coordinates of $\R^5$.
The impulse is located on the null hypersurface $\{U=0\}$ given by
 \begin{equation}
 {Z_2}^2+{Z_3}^2+\sigma{Z_4}^2=\sigma a^2\,,\label{surface}
 \end{equation}
which is a {\it nonexpanding} 2-sphere in the de~Sitter universe ($\Lambda>0$) and a
hyperboloidal 2-surface in the anti-de~Sitter universe ($\Lambda<0$), respectively.
Various 4-dimensional coordinate parametrizations of these spacetimes have been considered
e.g.\ in \cite{[B5]}. Physically the spacetime \eqref{classical}, \eqref{const} describes
impulsive gravitational waves as well as impulses of null matter.
Purely gravitational waves occur in case the vacuum field equations are satisfied.
It was demonstrated in  \cite{[B6],[B7]} that such solutions can be explicitly written as
\begin{equation}\label{eq:poles}
H(z,\phi)= \textstyle{ b_0\,Q_1(z)+\sum_{m=1}^\infty
b_m\,Q^m_1(z)\cos[m(\phi-\phi_m)]}\,,
\end{equation}
where  $z=Z_4/a$, $\tan\phi=Z_3/Z_2$ and $Q^m_1(z)$ are associated Legendre functions of the second kind generated
by the relation
${Q^m_1(z)=(-\sigma)^m|1-z^2|^{m/2}\frac{\d^m}{\d z^m}Q_1(z)}$.
The first term for ${m=0}$, i.e., ${Q_1(z)=\frac{z}{2}\log\left|\frac{1+z}{1-z}\right|-1}$,
represents the simplest axisymmetric Hotta--Tanaka solution (\cite{HotTan93}).
The components with ${m\ge1}$ describe nonexpanding impulsive gravitational
waves in (anti-)de~Sitter universe generated by null point sources with an
$m$-pole structure, localized on  the wave-front at the singularities
${z=\pm1}$. See \cite{[B7],[B6],Pod2002b} for more details.
\medskip

Now, the geodesics $\gamma$ of $M$ with tangent $T$ are characterized by
the condition that their $\overline M$-acceleration $A=\overline{\nabla}_{T}T$ is everywhere
normal to $M$. Denoting by $N$ the normal vector to
$M$ in $\overline M$ with $g(N,N)=\sigma$, we hence obtain
\begin{equation}
\overline{\nabla}_{T}T=-\sigma\, g(T,\overline{\nabla}_{T}N ) N.
\end{equation}
Using this identity and the constraint \eqref{const} the explicit form of the geodesic equation
was given in \cite[eq.\ (28)]{PO:01} as
\begin{align}\label{d'-geo-eq}
 \ddot U
  &= -\frac{1}{3}\,\Lambda\, U\, e\,,\nonumber\\
 \ddot V -
    \frac{1}{2}\, H\, \delta'\, \dot U^2
    -\delta^{pq}\,H_{,p}\, \delta\, \dot Z_q\, \dot U
  &=-\frac{1}{3}\,\Lambda\, V\, \big(e+\frac{1}{2}\,G\,\delta\, \dot U^2\big)
  \,,\nonumber\\
 \ddot Z_i -
     \frac{1}{2}\, H_{,i}\, \delta\, \dot U^2
  &=-\frac{1}{3}\,\Lambda\, Z_i\, \big(e+\frac{1}{2}\,G\,\delta\, \dot U^2\big)
  \,,\\
 \ddot Z_4 -
     \frac{\sigma}{2}\, H_{,4}\, \delta\, \dot U^2
  &=-\frac{1}{3}\,\Lambda\, Z_4\, \big(e+\frac{1}{2}\,G\,\delta\, \dot U^2\big)
  \,,\nonumber
\end{align}
where
\begin{equation}\label{eq:G}
 G:=\delta^{pq}\,  Z_p\, H_{,q}  - H,\quad \text{and}\quad e:=g(T,T)=\pm 1,0
\end{equation}
denotes the normalisation constant for spacelike (${e=1}$), timelike (${e=-1}$) and null (${e=0}$) geodesics
respectively. Observe, that  $\dot{}$ denotes the derivative with respect to an affine parameter $t$ which
we have suppressed in the equations. Here and in the following we also adopt the convention that Greek indices
$\alpha,\beta$ take all values $0,1,2,3,4$ while the indices $p,q, r$ are restricted to the
values $2,3,4$ and $i,j$ run from $2$ to $3$ only.

Obviously, these equations
reduce to the geodesic equations of the (anti-)de Sitter background off the impulse located at
$\{U=0\}$. Also observe that the first equation decouples from the rest of the
system and can be easily integrated. Consequently $U$ can be used as a parameter
of the remaining equations, a fact which is essential for the analysis
of the system \eqref{d'-geo-eq} presented in \cite[Sec.\ 4]{PO:01}. In fact,
there the geodesics of the (anti-)de Sitter background in front and behind the wave impulsive are
matched using a natural ansatz for solutions in the entire spacetime.
However, this procedure has to be viewed as being only heuristic since the
solution's $Z_p$-components are (assumed to be) continuous but not $\Con^1$, while
the $V$-component is (assumed to be) even discontinuous. Consequently the
solutions cannot be plugged back into the equations due to the occurrence of
undefined products of distributions, and so the question arises in which
sense they actually solve the equations, see the discussion at the end of Sec.\
4 of \cite{PO:01}. The situation is similar to the one encountered for impulsive \emph{pp}-waves
with $\Lambda=0$ and we refer to the discussion in \cite[Sec.\ II]{Steina} as well as to the general
discussion in \cite{H:85}.
\medskip

To resolve this open problem we now employ a regularisation approach and detail the setting we are working with:
To begin with we replace the  Dirac-$\delta$ in the line element \eqref{classical} by a fairly general
class of smooth approximations called \emph{model delta nets}. Chose an arbitrary smooth function
$\rho$ on $\R$ with unit integral and its support contained in $[-1,1]$. Then for $0<\eps\leq 1$ set
\begin{equation}\label{eq:delta}
\delta_\eps(x):=\frac{1}{\eps}\,\rho\Big(\frac{x}{\eps}\Big)\,.
\end{equation}
We now consider for fixed $\eps\in(0,1]$ the five-dimensional sandwich wave
\begin{equation}\label{5ipp}
\d s_{\eps}^{2}=\d Z_{2}^{2}+\d Z_{3}^{2}+ \sigma \d Z_{4}^{2}-2\d U \d V
 +H(Z_{2},Z_{3},Z_{4})\delta_{\eps}(U) \d U^{2}\,,
\end{equation}
and define the spacetime of our interest as $(M,g_\eps)$ given by the constraint \eqref{const}, i.e.,
\begin{equation}\label{const2}
F(U,V,Z_2,Z_3,Z_4):=-2UV+Z_{2}^{2}+Z_{3}^{2}+ \sigma Z_{4}^{2}-\sigma a^{2}=0\,.
\end{equation}
Observe that while the differential $\d F=2(-V,-U,Z_2,Z_3,\sigma Z_4)$
is independent of $\eps$, the normal vector $N^\alpha_\eps:=g^{\alpha\beta}_\eps \d F_\beta$
depends on $\eps$. Indeed we choose to work with the \emph{non-normalised} normal vector
$N_\eps$ to $M$ given by
\begin{equation}\label{eq:nv}
 N_\eps=(U,V+HU\delta_{\eps}(U),Z_p) \quad\text{with}\quad
 g_\eps(N_\eps,N_\eps)=\sigma a^2-U^2H\delta_\eps(U)\,.
\end{equation}
The non-zero Christoffel symbols of \eqref{5ipp} are given by
\begin{align}
\Gamma^V_{\eps\,UU}&=-\frac{1}{2}H\delta^{'}_{\eps}(U)\,,\quad
\Gamma^V_{\eps\,Up}=-\frac{1}{2}H_{,p}\delta_{\eps}(U)\,,\\
\Gamma^i_{\eps\,UU}&=-\frac{1}{2}H_{,i}\delta_{\eps}(U)\,,\quad
\Gamma^4_{\eps\,UU}=-\frac{1}{2}\sigma H_{,4}\delta_{\eps}(U)\,,
\end{align}
and the geodesics of $(M,g_\eps)$ are now characterised by
\begin{equation}
\overline{\nabla}^\eps_{T_{\eps}} T_{\eps}=-g_\eps(T_{\eps},\overline{\nabla}^\eps_{T_{\eps}} N_{\eps})\
\frac{N_{\eps}}{g_\eps( N_{\eps},N_{\eps})}\,,
\end{equation}
where (suppressing the parameter) we write
$\gamma_\eps=(U_\eps,V_\eps,Z_{p\eps})$ for the geodesics with tangent
$T_\eps=(\dot{U}_\eps,\dot{V_\eps},\dot Z_{p\eps})$ and $\overline{\nabla}^\eps$ denotes
the Levi-Civita connection of \eqref{5ipp}. By a straightforward calculation
we now obtain
\begin{equation}
 g_\eps(T_{\eps},\overline{\nabla}^\eps_{T_{\eps}} N_{\eps})
 =e+\frac{1}{2}\,\dot{U}_\eps^2\,\tilde{G_\eps}(U_\eps,Z_{p\eps})
 -\dot{U}_\eps\Big(H(Z_{p\eps})\,\delta_{\eps}(U_\eps)\,U_\eps\dot{\Big)},
\end{equation}
where we have used the abbreviations
\begin{equation}
 \tilde{G}_\eps(U,Z_r)
 :=\delta^{pq}\,H_{,p}(Z_r)\,\delta_{\eps}(U)\,Z_q
  + H(Z_r)\,\delta'_\eps(U)\,U\,,
 \quad \text{and}\quad
  e:=g_\eps(T_{\eps},T_{\eps})=\pm 1,0\,.
\end{equation}
Observe that since the $g_\eps$-norm of the tangent vector $T_\eps$ is constant
along the geodesic $\gamma_\eps$, we have chosen it also to be constant in $\eps$, which means
that we have fixed the normalisation independently of $\eps$.
Finally we obtain the following explicit form of the geodesic equations
\begin{align}\label{eq:geos}
\ddot{U}_\eps
&=-\Big( e + \frac{1}{2}\,\dot{U}_\eps^2\,\tilde{G_\eps}
           - \dot{U}_\eps\,\big(H
           \,\delta_{\eps}
           \,U_\eps\dot{\big)}\Big)\
           \frac{U_\eps}{\sigma a^2-U_\eps^2H\delta_\eps}\,, \nonumber\\
\ddot{V}_\eps-\frac{1}{2}\,H
           \,\delta^{'}_{\eps}
           \,\dot{U}_\eps^2 - \delta^{pq}H_{,p}
           \,\delta_{\eps}
           \,\dot{Z}^\eps_q\,\dot{U}_\eps
&=-\Big(e + \frac{1}{2}\,\dot{U}_\eps^2\,\tilde{G}_\eps
          - \dot{U}_\eps\,
            \big( H\, \delta_\eps\, U_\eps \dot{\big)}\Big)\
   \frac{V_\eps+H\,\delta_{\eps}U_\eps}
        {\sigma a^2-U_\eps^2H\delta_\eps}\,,\nonumber\\
\ddot{Z}_{i\eps}-\frac{1}{2}H_{,i}\,\delta_{\eps}\dot{U}_\eps^2
&=-\Big(e + \frac{1}{2}\,\dot{U}_\eps^2\,\tilde{G}_\eps
          - \dot{U}_\eps\,
            \big( H\, \delta_\eps\, U_\eps \dot{\big)}\Big)\
   \frac{Z_{i\eps}}{\sigma a^2-U_\eps^2H\delta_\eps}\,,\\
\ddot{Z}_{4\eps}-\frac{\sigma}{2}\,H_{,4}\,\delta_{\eps}\dot{U}_\eps^2
&=-\Big(e + \frac{1}{2}\,\dot{U}_\eps^2\,\tilde{G}_\eps
          - \dot{U}_\eps\,
            \big( H\, \delta_\eps\, U_\eps \dot{\big)}\Big)\
    \frac{Z_{4\eps}}{\sigma a^2-U_\eps^2H\delta_\eps}\,, \nonumber
\end{align}
where we again have suppressed the parameter $t$ as well as the dependencies on the
variables. However, note that always
\begin{align}\label{eq:ucomp}
 \delta_\eps&=\delta_\eps(U_\eps(t))\,,\quad \delta'_\eps=\delta'_\eps(U_\eps(t))\,,\nonumber \\
 \tilde G_\eps&=\tilde G_\eps(U_\eps(t),Z_{p\eps}(t))\,,  \quad
  H=H(Z_{p\eps}(t))\,,\quad \text{and}\quad H_{,p}=H_{,p}(Z_{q\eps}(t))\,.
\end{align}

The right hand sides of these equations are considerably more complicated than their
`distributional counterparts' in \eqref{d'-geo-eq},
the reason being that in the regularised equations the distributional identities $\delta(U)U=0$
and $\delta'(U)U=-\delta(U)$ do not apply. Indeed the lack of the first one leads to the more complicated
form of the normal vector (see \eqref{eq:nv}) and is reflected in the second summand in the
denominator of \eqref{eq:geos}.
On the other hand the lack of the second one is responsible for the different form of the terms
involving $\tilde G_\eps$ as compared to the ones involving $G$ (cf.\ \eqref{eq:G}) in \eqref{d'-geo-eq},
to which they reduce in
the limit $\eps\to 0$. Finally the terms proportional $(H\delta_\eps U_\eps)\dot{}$ which occur in all four
equations vanish in the limit $\eps\to 0$ again due to the first identity. The
same holds true for the
term proportional to $\tilde G_\eps U_\eps$ contained in the first equation.
In this sense the equations \eqref{eq:geos} converge weakly to \eqref{d'-geo-eq} as $\eps\to 0$.

The more complicated form of the system \eqref{eq:geos}, in particular, results in the fact that
the $U_\eps$-equation does \emph{not decouple} from the rest of the system and consequently $U$ cannot be used
as a parameter along the geodesics. This issue greatly complicates our analysis. However,
the $V_\eps$-equation still is linear and decoupled, hence can be simply
integrated once the rest of the system is solved.

\section{Existence, uniqueness and completeness of geodesics}\label{sec3}
In this section we prove an existence and uniqueness result for solutions
of the initial value problem for \eqref{eq:geos} that additionally guarantees that the
geodesics that enter the sandwich region $U\in[-\eps,\eps]$ at one side exist
long enough to leave it on the other side. This allows us to obtain global
solutions of the geodesic equations, since outside the strip $\{-\eps\leq U\leq
\eps\}$ the spacetime coincides with the background (anti-)de Sitter universe.

Observe that for fixed $\eps$ the equations are smooth and hence a local
solution is guaranteed to exist. However, the time of existence might depend on
$\eps$ and in principle could even shrink to zero as $\eps\to 0$. So the main
objective here is to provide a result which guarantees that the time of
existence is independent of $\eps$, and large enough such that the solutions
pass through the regularisation sandwich (at least for all small $\eps$). To
this end we employ a fixed point argument in the spirit of \cite[Appendix
A]{SS12} based on Weissinger's fixed point theorem (\cite{Wei:52}). However, the
significant increase in the complexity of the equations forces the use of new
ideas to derive the required estimates. In particular, since it is not possible
to use the $U$-coordinate as a parameter along the geodesics, the `singular
terms' such as $\delta_\eps$ are composed with the $U$-coordinate of the
solution $U_\eps$, see \eqref{eq:ucomp}. We have separated the technical proofs
preparing the grounds for the application of the fixed point theorem from the
main line of arguments and have deferred them to Appendix~\ref{A}.
\medskip

Let us start by giving the general setup. Consider any geodesic
\begin{equation}\label{eq:gamma}
 \gamma=(U,V,Z_p)
\end{equation}
on the background (anti-)de Sitter universe \emph{without} impulsive wave but
reaching $U=0$. All other geodesics are not of interest to the present analysis
and will be dealt with separately. We choose an affine parameter $t$ in such a
way that $U(t=0)=0$ and assume $\dot \gamma$ to be normalised by $e=\pm 1,0$.
Such geodesics can explicitly be written as
\begin{equation}\label{eq:backgrdgeo}
 U=t\,,\qquad U=a\dot U^0\sinh(t/a)\,,\qquad U=a\dot U^0\sin(t/a)\,,
\end{equation}
in the cases $\sigma e=0$, $\sigma e<0$, and  $\sigma e>0$, respectively, see
\cite[eq.\ (33)]{PO:01}. Recall that the case $\sigma e=0$ corresponds to null
geodesics in both de Sitter and anti-de Sitter space, while the case $\sigma
e<0$ corresponds to timelike geodesics in de Sitter as well as  to spacelike
geodesics in anti-de Sitter spacetime, and finally $\sigma e>0$ describes
spacelike geodesics  in de Sitter and timelike geodesics in anti-de Sitter
space. Without loss of generality we assume the constant $\dot U^0$ to be
positive\footnote{The time reversed case with $\dot{U}^0<0$ can be
treated in complete analogy.} so that in all three cases $U$ is increasing (at
least for $t\in[-a\pi/2,a\pi/2]$). It is thus most convenient to prescribe
initial data at $t=0$, that is we set
\begin{equation}\label{eq:data0}
 \gamma(t=0)=:(0,V^0,Z^0_p)\,,\qquad \dot \gamma(t=0)=:(\dot U^0>0,\dot V^0,\dot Z^0_p)\,,
\end{equation}
where the constants satisfy the constraints
\begin{equation}\label{const1}
(Z^0_{2})^{2}+(Z^0_{3})^{2}+ \sigma (Z^0_{4})^{2}-2U^0V^0=\sigma a^{2},\quad
Z^0_{2} \dot Z^0_{2}+Z^0_{3} \dot Z^0_{3}+ \sigma Z_{4}^0\dot  Z_{4}^0-V^0\dot U^0-U^0\dot V^0=0\,,
\end{equation}
which are simply consequences of the fact that we are dealing with $\gamma$ on
the (anti-)de Sitter manifold, see \eqref{const}. Note however that $U^0=0$, so the last term on
the left hand side of either condition actually vanishes. In addition the normalisation
condition
\begin{equation}\label{eq:norm}
 -2 \dot U^0\dot V^0+(\dot Z^0_2)^2+(\dot Z^0_3)^2+\sigma(\dot Z^0_4)^2=e
\end{equation}
holds.

Now we start to think of $\gamma$ as a geodesic in the impulsive wave spacetime
\eqref{classical}, \eqref{const} \emph{`in front' of the impulse} that is for
$U<0$. Also, $\gamma$ is a geodesic in the regularised spacetime \eqref{5ipp},
\eqref{const2} \emph{`in front' of the sandwich wave}, that is for $U\leq
-\eps$. We will call it \emph{`seed geodesic'} and denote the affine parameter time
when $\gamma$ enters this regularisation wave region by $\alpha_\eps$,
\begin{equation}
 U(t=\alpha_\eps)=-\eps\,.
\end{equation}
Observe that, by continuity of $\gamma$, $\alpha_\eps\to 0$ from below
as $\eps\to 0$. More precisely, we have
$\alpha_\eps=-\eps$, $\alpha_\eps=-a\mbox{Arcsinh}(\eps/a\dot U^0)$, and
$\alpha_\eps=-a\arcsin(\eps/a\dot U^0)$, respectively for the three cases in
\eqref{eq:backgrdgeo}, leading to $\alpha_\eps=-\eps/\dot U^0+O(\eps^3)$ in the
latter cases and hence overall
\begin{equation}\label{eq:aeps}
 -C\eps\leq \alpha_\eps<0\,,
\end{equation}
for some positive constant $C$.

\begin{figure}[h]
\begin{center}
\def\svgwidth{0.73\textwidth}
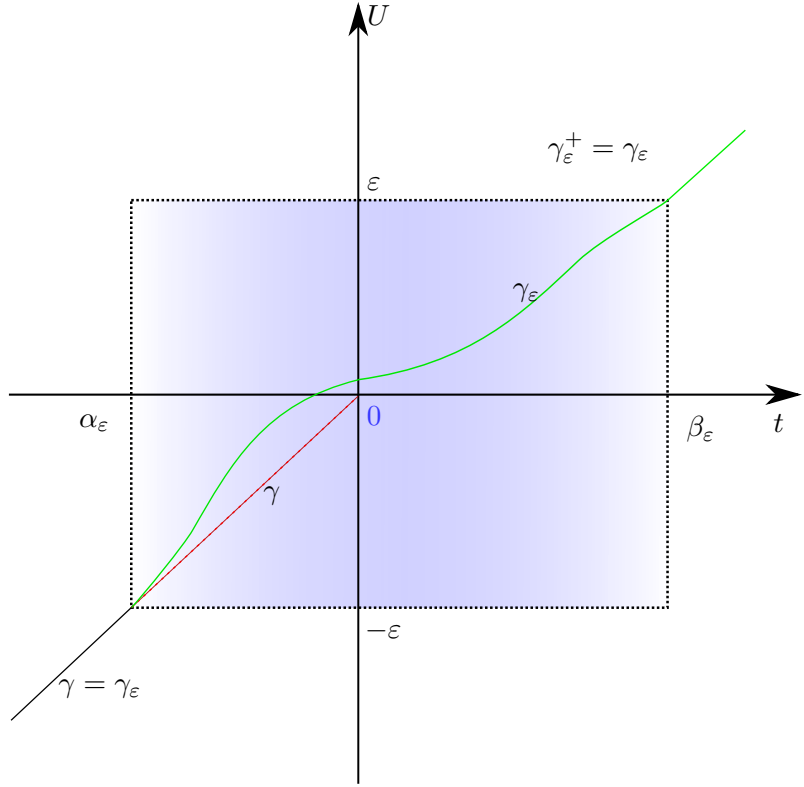
\caption{The $U$-component of the `seed geodesic' $\gamma$
is depicted in black until it reaches the regularisation sandwich at parameter
time $t=\alpha_\eps$, i.e., $U(\alpha_\eps)=-\eps$. While in the background
spacetime it would continue as the dotted red line to $U=0$ at $t=0$, in the regularised spacetime
it continues as $\gamma_\veps$ of \eqref{eq:gamma_eps} (depicted in
green) solving the equations
\eqref{eq:geos} with data \eqref{eq:dataeps}. Theorem
\ref{thm:global} guarantees that $\gamma_\eps$ (for $\eps$ small) leaves the
regularisation sandwich at $t=\beta_\eps$ and continues as background
geodesic $\gamma^+_\eps$ of \eqref{eq:gamma+} with data \eqref{eq:data+}.}
\end{center}\label{fig1}
\end{figure}

To investigate the geodesics in the regularised spacetime \eqref{5ipp}, \eqref{const2}, which is the main
topic of this paper, we follow $\gamma$ up to the beginning of wave zone, i.e., up to $t= \alpha_\eps$, and then
extend it (smoothly) to a geodesic
\begin{equation}\label{eq:gamma_eps}
 \gamma_\eps=(U_\eps,V_\eps,Z_{p\eps})
\end{equation}
solving the regularised geodesic equations \eqref{eq:geos}. This means
$\gamma_\eps$ at $\alpha_\eps$ assumes the data (see Figure~1)
\begin{equation}\label{eq:dataeps}
 \gamma_\eps(\alpha_\eps):=\gamma(\alpha_\eps)=:(-\eps,V^0_\eps,Z^0_{p\eps})\,, \quad
 \dot \gamma_\eps(\alpha_\eps):=\dot\gamma(\alpha_\eps)=:(\dot U_\eps^0,\dot V^0_\eps,\dot Z^0_{p\eps})\,.
\end{equation}
Observe that by smoothness of $\gamma$ the data $\gamma_\eps(\alpha_\eps)$ and  $\dot \gamma_\eps(\alpha_\eps)$
converge to $\gamma(0)$ and $\dot \gamma(0)$, respectively, as $\eps\to 0$. In
fact, by a mean value argument and \eqref{eq:aeps} we even have
\begin{align}\label{eq:dataconv}\nonumber
 |(-\eps,V^0_\eps,Z^0_{p\eps})-(0,V^0,Z^0_p)|
  &\leq \sup_{\alpha_\eps\leq t\leq 0}\|\dot\gamma(t)\|_h\, |\alpha_\eps|
  \ \leq \ C\eps\,,\\
 |(\dot U_\eps^0,\dot V^0_\eps,\dot Z^0_{p\eps})- (\dot U^0,\dot V^0,\dot Z^0_p)|
  &\leq \sup_{\alpha_\eps\leq t\leq 0}\|\ddot\gamma(t)\|_h\, |\alpha_\eps|
 \ \leq\ C\eps\,,
\end{align}
where $h$ is any Riemannian background metric and $C$ again is a generic constant.
\medskip

Based on Theorem \ref{thm:exmat} proved in Appendix \ref{A} we may now state and
prove our main results on the existence, uniqueness and completeness of the
geodesics in the \emph{regularised} spacetime \eqref{5ipp}, \eqref{const2}:

%

\begin{theorem}[Existence and uniqueness]\label{thm:ex}
Consider the geodesic equations \eqref{eq:geos} with initial data \eqref{eq:dataeps}.
Then for all $\eps$ small enough (more precisely for all $\eps\leq\eps_0$,
where $\eps_0$ is constrained by \eqref{cond:eps0}), there exists a unique smooth solution
$\gamma_\eps=(U_\eps,V_\eps,Z_{p \eps})$ on
$[\alpha_\eps,\alpha_\eps+\eta]$, where $\eta$ is independent
of $\eps$ (and explicitly given by \eqref{eq:b}).
\end{theorem}

\begin{proof}
As noted in the appendix it suffices to first solve the (simplified) model
system \eqref{eq:math} for $(u_\eps,z_\eps)$.
Identifying $(U_\eps,Z_{p\eps})$ with
$(u_\eps,z_\eps)$ the initial data \eqref{eq:data0} and \eqref{eq:dataeps}
transfer to the data \eqref{eq:zdata}, \eqref{eq:zdata_eps}. Then \eqref{eq:dataconv}
implies \eqref{eq:udata} and Theorem \ref{thm:exmat} applies to provide a unique smooth
solution $(U_\eps, Z_{p\eps})$ of \eqref{eq:geos}, \eqref{eq:dataeps} on $[\alpha_\eps,\alpha_\eps+\eta]$,
with $\eta$ given by \eqref{eq:b}.

Finally we solve the linear equation for $V_\eps$ to obtain the claimed smooth solution
$\gamma_\eps=(U_\eps,V_\eps,Z_{p\eps})$ on $[\alpha_\eps,\alpha_\eps+\eta]$.
\end{proof}

Next we make sure that the solutions just obtained, which by construction enter
the wave zone at $U_\eps=-\eps$ at parameter time $t=\alpha_\eps$ with
positive speed $\dot U_\eps^0$, in fact \emph{leave the sandwich region}, that
is they reach $U_\eps=\eps$ within their time of existence $\eta$. Consequently,
they naturally extend to the background \emph{(anti-)de Sitter universe `behind'} the sandwich
region. Observe that here it is vital that $\eta$ in \eqref{eq:b} is independent of
$\eps$.

\begin{theorem}[Extension of geodesics]\label{thm:global}
The unique smooth geodesics $\gamma_\eps$ of Theorem \ref{thm:ex} extend
to geodesics of the background (anti-)de Sitter spacetime `behind' the
sandwich wave zone.
\end{theorem}

\begin{proof} Let $\gamma_\eps=(U_\eps,V_\eps,Z_{p\eps})$ be the unique solution
of \eqref{eq:geos}, \eqref{eq:dataeps} given
by Theorem \ref{thm:ex}. By the definition of the `solution space' $\mathfrak{X}_\eps$ (see
\eqref{space}) we obtain
\begin{equation}
U_\eps(\alpha_\eps+\eta) = -\eps + \int_{\alpha_\eps}^{\alpha_\eps+\eta} \dot{U}_\eps(s)\,\dd s
\geq -\eps + \frac{\eta}{2} \dot{U}^0\geq -\eps + 3\eps \geq  \eps\,,
\end{equation}
since $\eps\leq\eta\,\dot U^0\!/6$ by \eqref{cond:eps00}.

So for such $\eps$ the geodesic $\gamma_\eps$ leaves the wave zone and extends
to a geodesic of the background spacetime since the geodesic equations
\eqref{eq:geos} coincide with the geodesic equations of the background (anti-)de
Sitter spacetime for $U_\eps\geq\eps$.
\end{proof}

Recall that by construction the global geodesic
$\gamma_\eps=(U_\eps,V_\eps,Z_{p\eps})$ with data \eqref{eq:dataeps} of Theorem
\ref{thm:global} for $t\leq\alpha_\eps$, i.e.\ `in front' of the sandwich,
coincide for all (small) $\eps$ with the single `seed geodesic' $\gamma$ with data \eqref{eq:data0}.
However, `behind' the sandwich the geodesics $\gamma_\eps$ for each $\eps$ will coincide with a
\emph{different} geodesic of the \emph{background} spacetime.
To make this observation more precise, define the affine parameter time when the geodesic
$\gamma_\eps$ leaves the sandwich wave zone by $\beta_\eps$,
\begin{equation}
 U_\eps(t=\beta_\eps)=\eps\,.
\end{equation}
and denote the corresponding values of $\gamma_\eps$ at $\beta_\eps$ by
\begin{equation}\label{eq:data+}
 \gamma_\eps(\beta_\eps)=:(\eps,V^{0+}_\eps,Z^{0+}_{p\eps})\,,
 \quad
 \dot \gamma_\eps(\beta_\eps)=:(\dot U_\eps^{0+},\dot V^{0+}_\eps,\dot
Z^{0+}_{p\eps})\,.
\end{equation}
Then for $t\geq\beta_\eps$ the geodesic $\gamma_\eps$  will coincide with the
geodesic
\begin{equation}\label{eq:gamma+}
 \gamma_\eps^+=(U^+_\eps,V^+_\eps,Z^+_{p\eps})
\end{equation}
of the \emph{background} (anti-)de Sitter space with the data \eqref{eq:data+}, see Figure~1 and
also Figure~2. Observe that the data \eqref{eq:data+} is normalised
and constrained, more precisely we have:

\begin{remark}[Preservation of constraints and normalisation]\label{rem:c+n}
The fact that the data \eqref{eq:data0} of the `seed geodesic' $\gamma$ in
\eqref{eq:gamma} is constrained and normalised, i.e., it satisfies
\eqref{const1} and \eqref{eq:norm}, implies that also the data
\eqref{eq:dataeps} of $\gamma_\eps$ is constrained and normalised. Clearly these
conditions are propagated by $\gamma_\eps$ being a solution to \eqref{eq:geos}.
Moreover, at $t=\beta_\eps$ the regularised metric $g_\eps$ and the background metric $g$
coincide and so the data \eqref{eq:data+} of $\gamma_\eps^+$ is constrained and
normalised with respect to the background spacetime.

While the preservation of the constraints confirms the consistency of our construction,
the preservation of the normalisation, in particular, implies that the causal character
of $\gamma_\eps$ (and $\gamma_\eps^+$) is the same as the one of the `seed' $\gamma$.
\end{remark}

Note that the geodesic $\gamma^+_\eps$ `behind' the regularisation sandwich
depends on $\eps$ (only) via this initial data. Interestingly, as we will detail
in Section \ref{sec5}, for $t>0$ the family of geodesics $\gamma_\eps$ of the
regularised spacetime converges for $\eps\to 0$ to a unique geodesic $\gamma^+$
in the background with data given by the limits of \eqref{eq:data+}. This will
explicitly describe the effect of the \emph{impulsive} gravitational wave on the
geodesics in (anti-)de Sitter universe.

\begin{figure}[h]
\begin{center}
\def\svgwidth{0.73\textwidth} 
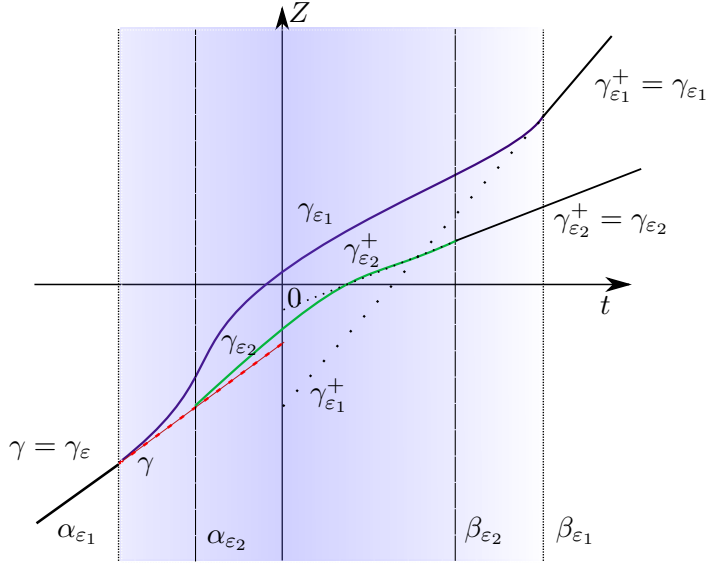
\caption{The $Z$-components of two solutions $\gamma_{\eps_1}$ (purple) and
$\gamma_{\eps_2}$ (green) of the regularised equations \eqref{eq:geos}
with the same `seed geodesic' $\gamma$ of \eqref{eq:gamma} are depicted
for $\eps_1>\eps_2$. The regularisation sandwich is given by
$[\alpha_{\eps_1},\beta_{\eps_1}]$ and $[\alpha_{\eps_2},\beta_{\eps_2}]$,
respectively. The dotted red line represents the $Z$-components of $\gamma$, while
the black dotted lines are said components of $\gamma_{\veps_i}^+$.}
\end{center}\label{fig2}
\end{figure}

\medskip
In the remainder of this section we will formulate completeness
results for the regularised spacetimes. First we remark that actually our
results allow us to make the smallness assumption on $\eps$ precise: $\eps$ has
to be smaller than $\eps_0$ constrained by \eqref{cond:eps0}.
This, however, means that the specific $\eps$ from which on a
certain geodesic $\gamma_\eps$ becomes complete depends on its data
\eqref{eq:data0}, i.e., on the 'seed geodesic', and there is in general \emph{no global $\eps$} from which on
\emph{all} geodesics hence the spacetime is complete.
A `global' completeness result in the spirit of \cite{SS14}, however, can be
obtained using the geometric theory of generalised functions (\cite[Ch.\
3]{GKOS:01}) and we reserve its detailed presentation for future work.

To formulate completeness results in our current setting, the dependence of
$\eps$ on the data discussed above also makes it necessary to be careful about
global effects. Indeed, geodesics in the background spacetimes with $\sigma
e>0$, that is spacelike geodesics in de Sitter space and timelike geodesics in
anti-de Sitter space are periodic. Consequently geodesics $\gamma_\eps$ in the
regularised spacetimes constructed from such `seed geodesics' $\gamma$ will
share their causal character (see Remark \ref{rem:c+n}) and hence cross the wave zone
infinitely often and we (have to) (re)apply Theorems \ref{thm:ex}, \ref{thm:global}
again and again. However, note that the geodesics may enter
the regularisation sandwich region each time with different data. So the $\eps$ from
which on Theorem \ref{thm:global} applies may in principal become smaller and
smaller on successive crossings with no positive infimum. In
such a case, given an initial geodesic $\gamma$ as in \eqref{eq:gamma} and given
any finite number $N$ of crossings we can specify an $\eps$ from which on the
geodesic $\gamma_\eps$ extends to cross the wave zone $N$-times. However, we
cannot give a (global) $\eps$ for which the geodesic $\gamma_\eps$ extends to
all (positive) values of its affine parameter. Consequently we prefer to avoid
multiple crossings of the impulse in the formulation of our results by
restricting to causal geodesics in the de Sitter-case (neglecting
unphysical tachyons only) and working with the universal covering spacetime in case
of anti-de Sitter space.

Note also that in our discussion so far (see the specification of the `seed
geodesics' $\gamma$ at the beginning  of this section) we have exclusively dealt
with geodesics  with non-constant $U$-component. Hence it remains to deal with
geodesics travelling parallel to the surface $\{U=0\}$. In case
$U=\mbox{const}\not=0$ the geodesic will never enter the sandwich region of the
regularised spacetime, provided $\eps$ is small enough. Staying entirely on the
constant curvature background such a geodesic clearly is complete. To discuss
the geodesics with $U=0$, observe that the surface $\{U=0\}$ is totally geodesic
(not only in the background but also) in the regularised spacetime, which can be
seen from the $U$-component of the geodesic equations \eqref{eq:geos} (see also
\cite{PSSS:15}, the discussion prior to Thm.\ 3.6). Hence such geodesics
have trivial $U$-components and consequently  the system \eqref{eq:geos}
reduces to the background geodesic equations which again leads to completeness.
So we finally arrive at:

\begin{theorem}[Causal completeness for positive $\Lambda$]\label{thm:compl_plus}
Every causal geodesic in the entire class of regularised  nonexpanding
impulsive gravitational waves propagating in de Sitter universe
(i.e., \eqref{5ipp}, \eqref{const2} with $\Lambda>0$ and a smooth profile
function $H$) is complete, provided the regularisation parameter $\eps$ is
chosen small enough.
\end{theorem}

\begin{theorem}[Completeness for negative $\Lambda$]\label{thm:compl_minus}
Every geodesic in the entire class of regularised  nonexpanding impulsive
gravitational waves propagating in the universal cover of anti-de Sitter
universe (i.e., \eqref{5ipp}, \eqref{const2}, with $\Lambda<0$ and a smooth
profile function $H$) is complete, provided the
regularisation parameter $\eps$ is chosen small enough.
\end{theorem}

\begin{remark}[Non-smooth profiles $H$]
In case the profile function $H$ in the metric \eqref{5ipp} is
non-smooth --- which, in fact, occurs in physically interesting models where $H$ possesses
poles on the wave front at ${z=\pm 1}$, see \eqref{eq:poles} --- our method
still applies but some care is needed. Indeed, if a `seed geodesic' $\gamma$
hits the surface at ${U=0}$ away from
the poles we may work on an open subset of the spacetime with
the poles of $H$ removed. We only have to choose the constant $C_1$ in
\eqref{space} so small that the curves in the `solution space'
$\mathfrak{X}_\eps$ stay away from the poles, and then Theorems \ref{thm:ex}
and \ref{thm:global} still apply. The only `seed geodesics' $\gamma$ which do
\emph{not} allow for such a treatment are those which \emph{directly head} at
the poles (i.e., $\gamma$ of \eqref{eq:gamma} hits the pole at $\gamma(0)$),
which is in complete agreement with physical expectations.
\end{remark}

\section{Boundedness properties of geodesics}\label{sec4}

In this section we establish boundedness properties of
the global geodesics in the regularised spacetimes obtained in the previous
section. In particular, we will prove local boundedness of $\gamma_\eps$ and of
some components of its velocity \emph{uniformly in $\eps$}. These properties
will be essential in the next section where we derive the limits of
$\gamma_\eps$.

To begin with observe that the fixed point argument of the appendix already
gives uniform boundedness of the $U$- and $Z_p$-components together with their
first order derivatives. On the other hand, the $V$-component was not involved in the fixed
point argument and we have to establish its boundedness properties using the
$V$-component of the geodesic equation \eqref{eq:geos}. Since this equation
involves a $\delta'$-term which is \emph{not} multiplied by $U_\eps$, the
$V$-component of $\dot \gamma_\eps$ is not uniformly bounded in the
regularisation sandwich. However, we will show that $\dot V_\eps(\beta_\eps)$,
i.e., the $V$-speed when the geodesic \emph{leaves} the regularisation sandwich
is uniformly bounded.

\begin{proposition}[Uniform boundedness of geodesics]\label{prop:bounds}
The global geodesics $\gamma_\eps=(U_\eps,V_\eps,Z_{p\eps})$ of Theorem
\ref{thm:global} satisfy
\begin{enumerate}
 \item $U_\eps$ and $\dot U_\eps$ are locally uniformly bounded in $\eps$,
 \item $Z_{p\eps}$ and $\dot Z_{p\eps}$ are locally uniformly bounded in
$\eps$,
 \item $V_\eps$ is locally uniformly bounded in $\eps$, and
 \item $\dot V_\eps(\beta_\eps)$ is uniformly bounded in $\eps$.
\end{enumerate}
\end{proposition}


Observe that by Lemma \ref{lem:gamma} the time $\beta_\eps$ when
the geodesic leaves the regularisation strip  satisfies
$\beta_\ep\leq \alpha_\eps+4\ep/\dot U^0\to 0$.

\begin{proof}
Items (i) and (ii) are immediate from Theorem \ref{thm:exmat}.

To deal with the $V$-component we write
\begin{equation}
 |V_\eps(t)|\leq|V^0|+|\dot V^0|\,|t-\alpha_\eps|
                  +\int\limits_{\alpha_\eps}^t\int\limits_{\alpha_\eps}^s|\ddot
V_\eps(r)| \,\d r \d s
\end{equation}
and estimate each term in the differential equation \eqref{eq:geos} for
$V_\eps$. To begin with we claim that
\begin{equation}\label{claim}
 \mbox{$V_\eps$ is bounded on $[\alpha_\eps,\beta_\eps]$}\,.
\end{equation}
Indeed, proceeding as in the proof of Proposition \ref{self} and, in particular,
using \eqref{space} and Lemmata \ref{lem:denom},\ref{lem:gamma} we obtain
\begin{align}\label{eq:Vepsteps}\nonumber
 \int\limits^{t}_{\alpha_\eps}\int\limits_{\alpha_\eps}^s|\ddot V_\eps(r)| \,\d r \d s
 \leq&\  \Big(\frac{4\eps}{\dot U^0}\Big)^2
   \left(\frac{1}{2}\|H\|_\infty\frac{1}{\eps^2}\|\rho'\|_\infty
         \frac{9}{4}(\dot{U}^0)^2
    +3\|DH\|_\infty\frac{1}{\eps}\|\rho\|_\infty\big(|\dot Z^0|+C_2\big)
         \frac{3}{2}\dot U^0\right)
\nonumber \\
  &+\frac{2}{a^2}\Big(\frac{4\eps}{\dot
   U^0}\Big)^2
  \left(1+\frac{9}{8}(\dot{U}^0)^2\|\tilde G_\eps\|_\infty+\frac{3}{2}\dot U^0
  \|(H\delta_\eps U_\eps\dot{)} \|_\infty\right)
  \|H\|_\infty \|\rho\|_\infty
\nonumber \\
 &+\frac{2}{a^2}\left(1+\frac{9}{8}(\dot{U}^0)^2\|\tilde
  G_\eps\|_\infty+\frac{3}{2}\dot U^0 \|(H\delta_\eps U_\eps\dot{)}\|_\infty
  \right)\int\limits^{t}_{\alpha_\eps} \int\limits_{\alpha_\eps}^s|V_\eps(r)| \,\d r \d s
\\ \nonumber
 &\leq C+ C\int\limits^{t}_{\alpha_\eps}
\int\limits_{\alpha_\eps}^s|V_\eps(r)|\left(1+\frac{\chi_{\eps}(r)}{\eps}
\right) \,\d r \d s\,,
\end{align}
where $C$ is some generic constant and $\chi_{\eps}$ is the
characteristic function of $[\alpha_\eps,\beta_\eps]$. Moreover, we have used
that $\|\tilde
G_\eps\|_\infty=O(1/\eps)= \|(H\delta_\eps U_\eps\dot{)}\|_\infty$. So overall
we obtain by a generalization of Gronwall's inequality due to Bykov
\cite[Thm.11.1]{BS:92}\footnote{Note that it suffices that $\chi_\eps$ is
integrable rather than continuous.}
\begin{equation}
 |V_\eps|
   \leq C(1+\eps)\exp\Big(\int\limits^{\beta_\eps}_{\alpha_\eps}
   \int\limits_{\alpha_\eps}^s C\big(1+\frac{1}{\eps}\big) \d r \d s\Big)
  \leq C\, e^C\,,
\end{equation}
establishing the claim \eqref{claim}.
\medskip

Now we may prove (iv): Writing $\dot V_\eps(\beta_\eps)$ also as an integral and
using boundedness of $V_\eps$ on $[\alpha_\eps,\beta_\eps]$ we may proceed as in
\eqref{eq:Vepsteps}. Then again we obtain uniform boundedness of
all the terms but the first one, which involves the $\delta'_\eps$-term. To
estimate this one we use integration by parts to obtain
\begin{align}
 \int_{\alpha_\eps}^{\beta_\eps} &H(Z_\eps(s))\,\delta_\eps'(U_\eps(s))\, \dot
U_\eps^2(s)\,\d s
 = \int\limits_{-\eps}^\eps H\big(Z_\eps(U_\eps^{-1}(r))\big)\, \delta_\eps'(r)
   \dot U_\eps(U_\eps^{-1}(r))\,\d r
\nonumber \\
& = \delta_\eps(r)\, H\big(Z_\eps(U_\eps^{-1}(r))\big) \, \dot
U_\eps(U_\eps^{-1}(r))
 \Bigr|_{-\eps}^\eps-
\int\limits_{-\eps}^\eps\delta_\eps(r)\Big(H\big(Z_\eps(U_\eps^{-1}(r))\big) \,
\dot U_\eps\big(U_\eps^{-1}(r)\big)\dot{\Big)}\, \d r\,.
\end{align}
Now the first term vanishes and the second one is bounded independently of
$\eps$ by (i), (ii), the fact that $\dot U_\eps^{-1}$ is uniformly bounded
away from zero by \eqref{space}, and since $\ddot U_\eps$ is uniformly bounded
by \eqref{eq:geos}. This establishes (iv).
\medskip

To prove (iii) it remains to show that $V_\eps$ is bounded on
any compact interval disjoint from $(\alpha_\eps,\beta_\eps)$.
But this follows from the fact that $\gamma_\eps$ outside of
$(\alpha_\eps,\beta_\eps)$ solves the geodesic equation of the background
spacetime of constant curvature, and that $\gamma_\eps(\beta_\eps)$ and
$\dot\gamma_\eps(\beta_\eps)$ are uniformly bounded by (i), (ii), (iv) and
\eqref{claim}.
\end{proof}

\section{Limiting geodesics}\label{sec5}

In this final section we consider the limit $\eps\to 0$ of the unique global smooth
geodesics $\gamma_\eps$ of the regularised spacetime \eqref{5ipp},
\eqref{const2} obtained in Section \ref{sec3} (Thm.\ \ref{thm:ex}, Thm.\
\ref{thm:global}). This physically amounts to \emph{explicitly determine the
geodesics of the distributional form} \eqref{classical}, \eqref{const} of all
nonexpanding impulse gravitational waves propagating in (an\-ti{-)}de Sitter
universe. In particular, we will prove that the geodesics $\gamma_\eps$
converge to geodesics of the background \mbox{(anti-)}de Sitter
spacetime with appropriate \emph{but different} data on either side of the
impulse ($U<0$ and $U>0$ respectively), which from a global point of view
amounts to convergence of $\gamma_\eps$ to geodesics of the background which
have to be \emph{matched appropriately} across the impulse. The technical calculation
of the limits is given in Appendix \ref{B}.

To make our claims on the convergence of $\gamma_\eps=(U_\eps,V_\eps,Z_{p\eps})$
precise we introduce the following notation for the geodesics of the
background \mbox{(anti-)}de Sitter universe: Let $\gamma=(U,V,Z_{p})$ be a `seed geodesic'
as in \eqref{eq:gamma}, that is $U(t=0)=0$ and $\gamma$ assumes the
data \eqref{eq:data0}, i.e.,
\begin{equation}
 \gamma(0)=(0,V^{0},Z^{0}_p)\,,\quad \text{and}\quad
 \dot \gamma(0)=(\dot U^{0}>0,\dot V^{0},\dot Z^{0}_p)\,,
\end{equation}
where the constants satisfy the constraints \eqref{const1} and are
normalised as in \eqref{eq:norm}.
Furthermore let $\gamma^+=(U^+,V^+,Z^+_{p})$ be a geodesic of the
background again crossing $U=0$ at $t=0$, i.e., with $U^+(t=0)=0$ and with data
\begin{equation}\label{eq:Xdata_ro}
\gamma^+(0)=(0,{\sf B},Z_p^0)\,,\quad \text{and}\quad
\dot \gamma^+(0)=(\dot U^0,{\sf C},{\sf A_p})\,,
\end{equation}
where we define
\begin{equation}\label{eq:limconst}
 {\sf A_p}:=\lim_{\eps\to 0} \dot Z_{p\eps}(\beta_\eps)\,,\quad
    {\sf B}=\lim_{\eps\to 0} V_\eps(\beta_\eps)\,,\quad
    {\sf C}=\lim_{\eps\to 0} \dot V_\eps(\beta_\eps)\,.
\end{equation}
Recall that $\beta_\eps\leq \alpha_\eps+4\ep/\dot{U}^0\to 0$ is
defined to be the time when the regularised geodesic $\gamma_\eps$ leaves the
regularisation strip, i.e., $U_\ep(\beta_\ep)=\ep$. Finally define $\tilde\gamma=
(\tilde U,\tilde V,\tilde Z_p)$ by
\begin{equation}\label{eq:limitgeos}
 \tilde\gamma(t):=\begin{cases}
  \gamma(t)\,,\qquad&t \leq 0\\
  \gamma^+(t)\,,\qquad&t>0\,.\end{cases}
\end{equation}

We will show that $\gamma_\eps$ \emph{converge to the `matched geodesics'}
$\tilde\gamma$ of the impulsive spacetime, which from now on we will also call `limiting
geodesics' with `past branch' $\gamma$ and `future branch' $\gamma^+$, see also
Figure~3.
Note that the respective notion of convergence of the individual
components of $\gamma_\eps$ will differ, subject to the regularity of the
respective components of the `limiting geodesics'. Indeed,
$\tilde\gamma=(\tilde U,\tilde V,\tilde Z_p)$ has a smooth first component $U$,
while $Z_p$ is continuous
with a finite jump in $\dot{Z}_p$ (determined by $A_p$) across the impulse at
$t=0$, and $V$ is even discontinuous across $t=0$ with a finite jump in $V$ and
$\dot V$ (determined by the coefficients ${\sf B}$ and ${\sf C}$, respectively).

Observe that at the moment we only know the limits in
\eqref{eq:limconst} to exist for subsequences (by uniform boundedness of
$\dot Z_{p\eps}$, $V_\eps$ and $\dot V_\eps(\beta_\eps)$, cf.\ Proposition
\ref{prop:bounds}) and hence ${\sf A_p}$, ${\sf B}$, and ${\sf C}$ need not be uniquely
defined. We will, however, prove convergence and we will derive an explicit expressions for ${\sf A_p}$, ${\sf B}$ and ${\sf
C}$ in Proposition \ref{prop:explicit_limits}, below.
But first we state and prove the main assertion on the limits of the
geodesics in the regularised spacetime:

\begin{theorem}\label{prop:lim1}
The geodesics $\gamma_\eps=(U_\eps,V_\eps,Z_{p\eps})$
of the regularised spacetime derived in Theorem \ref{thm:global}
converge to the `limiting geodesics' $\tilde\gamma$ of \eqref{eq:limitgeos} in
the following sense:
\begin{enumerate}
 \item $U_\ep \to \tilde U$ in $\Con^1$,
 \item $Z_{p\ep}\to \tilde Z_p$ locally uniformly, \\
  $\dot Z_{p\ep}\to\dot{\tilde{Z}}_p$ in distributions and
  uniformly on compact intervals not containing $t=0$,
 \item $V_\eps\to \tilde V$ in distributions and
  in $\Con^1$ on compact intervals not containing $t=0$.
\end{enumerate}
\end{theorem}

Observe that $\dot{\tilde{Z}}_p$ is discontinuous across $t=0$, hence convergence
of $\dot Z_{p\eps}$  cannot be uniform on any interval containing
$t=0$\footnote{Unless $t=0$ is the right endpoint of the interval.
This, however, is rather an artefact due to our choice setting
$\tilde \gamma(0)=\gamma(0)$, cf.\ \eqref{eq:limitgeos}.}, and the same holds
true
for $V_\eps$ and $\dot V_\eps$.

\begin{proof}
First we consider the $U$-component of $\gamma_\eps$ on the interval $[\alpha_\eps,\beta_\eps]$, where
we have from the geodesic equations \eqref{eq:geos} resp.\
\eqref{d'-geo-eq}
\begin{align}\nonumber
 |U_\ep(t)-U(t)|
 \leq \ep  &+ |e|\int_{\alpha_\eps}^t\int_{\alpha_\eps}^s
\Bigl|\frac{U_\ep}{\sigma a^2 - U_\ep^2 H\delta_\ep} -
\frac{U}{\sigma a^2}\Bigr|\, \dd r\,\dd s \\
 &+ \int_{\alpha_\eps}^t \int_{\alpha_\eps}^s
\Bigl|\frac{\frac{1}{2}U_\eps\dot{U_\eps}^2\tilde{G}_\eps-U_\eps\dot{U_\eps}
\left(H\delta_{\veps}
U_\eps\right)\dot{}}{ \sigma a^2-U_\eps^2H\delta_{\veps}}\Bigr| \,\dd r\,\dd s
 =: \ep + |e|I + II\, .
\end{align}
To estimate $I$ observe that (cf.\ the proof of Lemma \ref{lem:denom})
\begin{align}\label{eq:estdenom}
 \Bigl| \frac{1}{\sigma a^2 - U_\ep^2 H \delta_\ep} - \frac{1}{\sigma a^2}\Bigr| \leq \frac{2}{a^4}\ep \|H\|_\infty
\|\rho\|_\infty \leq C \ep\,,
\end{align}
with $C$ a generic constant, and consequently by \eqref{space}
\begin{align}\nonumber
 I &\leq \int_{\alpha_\eps}^t\int_{\alpha_\eps}^s \left
(\Bigl|\frac{U_\ep}{\sigma a^2 - U_\ep^2 H\delta_\ep} -
\frac{U_\ep}{\sigma a^2}\Bigr| + \Bigl|\frac{U_\ep}{\sigma a^2 } - \frac{U}{\sigma a^2}\Bigr|\right) \dd r\,\dd s\\
&\leq \eta^2 (\ep + C_1) C \ep  + \frac{1}{a^2}
\int_{\alpha_\eps}^t\int_{\alpha_\eps}^s \Bigl|U_\ep - U \Bigr|\, \dd r\,\dd
s\,.
\end{align}
For the term $II$ we obtain, as in the proof of Proposition \ref{self} (cf.\ \eqref{eq:geps}, \eqref{eq:bracket}),
that $II \leq C \ep$. Hence, overall
\begin{equation}
 |U_\ep(t)-U(t)|\leq C \ep + \frac{|e|}{a^2}
\int_{\alpha_\eps}^t\int_{\alpha_\eps}^s \Bigl|U_\ep - U \Bigr|\, \dd r\,\dd
s\,,
\end{equation}
and so again by Bykov's inequality  $|U_\ep(t)-U(t)|=O(\eps)$.
In the same way we see that also
$|\dot{U}_\ep(t) - \dot{U}(t)|=O(\ep)$ and so
\begin{equation}\label{eq:Uunif}
 \sup_{\alpha_\eps\leq t\leq \beta_\eps} | U_\eps(t)-U(t)| + |\dot{U}_\eps(t) - \dot{U}(t)| \to 0\,.
\end{equation}

\medskip
We now turn to the $Z_p$-components on the interval $[\alpha_\eps,\beta_\eps]$.
We have
\begin{align}\label{eq:Zunif}
 \sup_{\alpha_\eps\leq t\leq \beta_\eps}|Z_{p\eps}(t)-\tilde Z_p(t)|
 \leq \sup_{\alpha_\eps\leq
t\leq\beta_\eps}|Z_{p\ep}(t)-Z_{p\eps}^0|+\sup_{\alpha_\eps\leq t\leq
\beta_\eps}|Z_{p\eps}^0-\tilde Z_p(t)|
 \to 0 \qquad (\eps\to 0)\,,
\end{align}
where $Z^0_{p\eps}=Z_p(\alpha_\eps)$ was defined in \eqref{eq:dataeps}.
Indeed by continuity of $\tilde Z_p$ the second term converges to zero since
$\alpha_\eps$, $\beta_\ep\to 0$ and the first term can be estimated using the
differential equation by
\begin{align}\label{eq:Z-est}\nonumber
 |Z_{p\ep}(t)-Z_{p\eps}^0|\leq&
\int_{\alpha_\eps}^{\beta_\ep}\int_{\alpha_\eps}^s
\left|\frac{DH\delta_\veps\dot{U}_\eps^2}{2}
  -\frac{eZ_p^\eps +\tfrac{1}{2}Z_p^\eps\dot{U_\eps}^2\tilde{G}_\eps
-Z_p^\eps\dot{U_\eps}\left(H\delta_{\veps} U_\eps\right)\dot { }}
   {\sigma a^2-U^2H\delta_{\veps}}\,\right| \dd r\,\dd s  \\
   &+|\dot{Z}_{p\eps}^0| (\alpha_\ep + \beta_\eps)\leq C \ep\,.
\end{align}
Here we have used that the inner integral is bounded by
\eqref{eq:iint1}-\eqref{eq:iint3}.
\medskip

Now we finish the proof of (i) and establish the claim on uniform convergence in
(ii) and (iii). First note that by construction there is nothing to show for
$t\leq 0$. For $t\geq\beta_\eps$ we use continuous dependence of solutions to
ODEs on the data. Indeed for such $t$ both $\gamma_\eps$ and $\tilde\gamma=\gamma^+$ are
solutions to the same differential equation, however, with different data which
is given for $\gamma_\eps$ at $t=\beta_\eps$ by \eqref{eq:Xdata_ro} and for
$\gamma^+$ at $t=0$ by \eqref{eq:data+}. More precisely, for all $T>0$
(which implies $T>\beta_\eps$ for $\eps$ small) we have
\begin{align}\nonumber\label{eq:max2}
 \sup_{\beta_\eps\leq t\leq
T}\big(|\gamma_\eps(t)-&\gamma^+(t)|,|\dot \gamma_\eps(t)-\dot
\gamma^+(t)|\big)\\
 &\leq \max \big(|\gamma_\eps(\beta_\eps) - \gamma^+(\beta_\eps)|,
|\dot{\gamma}_\ep(\beta_\eps) -
\dot{\gamma}^+(\beta_\eps)|\big)\, e^{TL}\,,
\end{align}
where $L$ is a Lipschitz constant for the right hand side of the geodesic
equation of the background on a suitable compact set. Note that such a set
exists by the boundedness properties of $\gamma_\eps$ established in
Proposition \ref{prop:bounds}, i.e., $\dot\gamma_\eps(\beta_\eps)$ is uniformly
bounded. Finally, for the terms in the maximum in
\eqref{eq:max2}  we have
\begin{align}
 &|U_\eps(\beta_\eps)-U^+(\beta_\eps)|\to 0\,,\quad
 |\dot U_\eps(\beta_\eps)-\dot U^+(\beta_\eps)|\to 0
  \quad \text{by \eqref{eq:Uunif}, and}\\\nonumber
 &|Z_{p\eps}(\beta_\eps)-Z^+_p(\beta_\eps)|\to 0
   \quad \text{by \eqref{eq:Zunif}}\,,
\end{align}
whereas for the remaining terms we write
\begin{align}\nonumber \label{eq:max3}
 |V_\eps(\beta_\eps)-V^+(\beta_\eps)|
  &\leq |V_\eps(\beta_\eps)-V^+(0)| +
|V^+(0)-V^+(\beta_\eps)|\,,\\
|\dot V_\eps(\beta_\eps)-\dot V^+(\beta_\eps)|
  &\leq |\dot V_\eps(\beta_\eps)-\dot V^+(0)| + |\dot
V^+(0)-\dot V^+(\beta_\eps)|\,,\\ \nonumber
|\dot Z_{p\eps}(\beta_\eps)-\dot Z_p^+(\beta_\eps)|
  &\leq |\dot Z_{p\eps}(\beta_\eps)-\dot Z_p^+(0)| + |\dot
Z_p^+(0)-\dot
   Z_p^+(\beta_\eps)|\,.
\end{align}
Now in each line the last term on the right hand side goes to zero by smoothness
of $\gamma^+$, while for the respective first terms we have by
our choice of data \eqref{eq:Xdata_ro}, \eqref{eq:limconst}
\begin{align}\label{eq:convdata}\nonumber
 |V_\eps(\beta_\eps)-V^+(0)|&=|V_\eps(\beta_\eps)-{\sf B}|\to 0\,,
\\
 |\dot V_\eps(\beta_\eps)-\dot V^+(0)|&=|\dot V_\eps(\beta_\eps)-{\sf C}|\to 0\,,
\\\nonumber
|\dot Z_{p\eps}(\beta_\eps)-\dot Z_p^+(0)|&=|\dot Z_{p\eps}(t_\eps)-{\sf A_p}|\to 0
\,.
\end{align}
\medskip

Finally to prove the distributional convergence in (iii) thanks to
the uniform convergence of $V_\eps$ established above we
only have to consider the integral
\begin{equation}
 \int_{\alpha_\eps}^{\beta_\eps} \big(V_\eps(s)-V(s)\big)\varphi(s) \,\d s
\end{equation}
for a test function $\varphi$ on $\R$. This, however, converges to zero
by the local uniform boundedness of $V_\eps$ established in Proposition
\ref{prop:bounds}(iii). In case of the distributional convergence in (ii) we
argue precisely in the same manner, now using Proposition \ref{prop:bounds}(ii).
\end{proof}

\begin{remark}[Normalisation and constraints in the limit]
The convergence result provided by Theorem \ref{prop:lim1} also
guarantees the preservation of the normalisation, i.e., the
normalisation of the `seed geodesic' $\gamma$, which by Remark \ref{rem:c+n}
carries over to the regularised geodesics $\gamma_\eps$ (and hence to
$\gamma_\eps^+$), also
carries over to the `future branch' of the `limiting geodesic' $\gamma^+$. To see this we
just write
\begin{align}
 e\ &=\ g_{\eps}(\gamma_\eps(\beta_\eps))\big(\dot\gamma_\eps(\beta_\eps),\dot\gamma_\eps(\beta_\eps)\big)
     = g\big(\dot\gamma_\eps(\beta_\eps),\dot\gamma_\eps(\beta_\eps)\big)
 \nonumber\\
   &=\ -2\dot U_\eps(\beta_\eps)\dot V_\eps(\beta_\eps)
        +(\dot{Z}_{2\eps}(\beta_\eps))^2+(\dot{Z}_{3\eps}(\beta_\eps))^2+\sigma (\dot{Z}_{4\eps}(\beta_\eps))^2
 \\\nonumber
   & \to\, -2\dot U^0 {\sf C} + {\sf A_2}^2+{\sf A_3}^2+\sigma {\sf A_4}^2 =
    g\big(\dot\gamma^+(0),\dot\gamma^+(0)\big).
\end{align}
Here the first equality follows from Remark \ref{rem:c+n} and the second one follows from the fact, that
at $\gamma_\eps(\beta_\eps)$ the regularised metric agrees with the (constant) background metric.
Finally, convergence is due to Thm.\ \ref{prop:lim1}(i) and our choice of the
data \eqref{eq:Xdata_ro},
\eqref{eq:limconst}.

Also by a similar (actually simpler) argument the constraints carry over from $\gamma_\eps$ to
$\gamma^+$, which again confirms consistency of our construction.
\end{remark}

To end this section and the entire paper we now explicitly evaluate the limits
in \eqref{eq:limconst}, thereby
showing that the `limiting geodesics' \eqref{eq:limitgeos} of the smooth global
geodesics $\gamma_\eps$ of Theorem \ref{thm:global} in the
regularised spacetime \eqref{5ipp}, \eqref{const2} coincide with the geodesics \eqref{d'-geo-eq}
of the distributional spacetime \eqref{classical}, \eqref{const} derived previously in \cite{PO:01}.

\begin{figure}[h]
\begin{center}
\def\svgwidth{0.73\textwidth} 
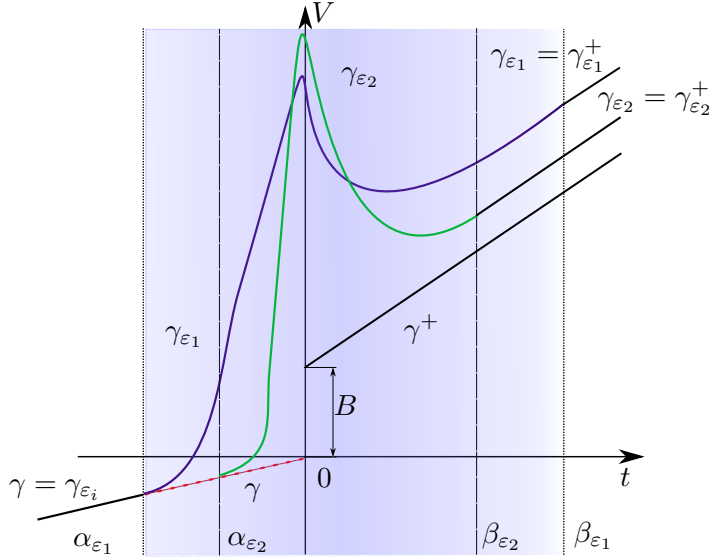
\caption{The $V$-components of $\gamma_{\eps_1}$ (purple) and $\gamma_{\eps_2}$
 (green) for $\eps_1>\eps_2$ are depicted. They converge to
 the  the `limiting geodesic' $\tilde\gamma$ whose `future branch' $\gamma^+$
 is separated from its `past branch' $\gamma$ (black outside and dotted red
 inside the regularisation sandwich) by the `jump' ${\sf B}$ calculated in Proposition
\ref{prop:explicit_limits}. The `jump' ${\sf C}$ in
 $\dot V$ is indicated by the different $V$-slopes  of the `past branch'
 $\gamma$ and the `future branch' $\gamma^+$.}
 \label{fig3}
\end{center}
\end{figure}

\begin{proposition}\label{prop:explicit_limits}
The constants ${\sf A_p}$, ${\sf B}$, and ${\sf C}$ determining the data for the `future branch'
$\gamma^+$ of the `limiting geodesics' $\tilde\gamma$ are explicitly given by
\begin{align}\label{eq:ro:consistency}\nonumber
 {\sf A_i} =\lim_{\eps\to 0}\dot Z_{i\eps}(\beta_\eps)\
 &=\ \frac{1}{2}\,\dot U^0\,
         \Bigl(\ H_{,i}(Z_r^0) + \frac{Z_i^0}{\sigma a^2}\big(H(Z_r^0)
          - \delta^{pq}Z_p^0H_{,q}(Z_r^0) \big)\Bigr)+\dot Z^0_i\,,\\
 {\sf A_4} =\lim_{\eps\to 0}\dot Z_{4\eps}(\beta_\eps)\
 &=\ \frac{1}{2}\,\dot U^0\,
         \Bigl(\sigma H_{,4}(Z_r^0) + \frac{Z_4^0}{\sigma a^2}\big(H(Z_r^0) -
 \delta^{pq}Z_p^0H_{,q}(Z_r^0) \big)\Bigr)+\dot Z^0_4 \,, \nonumber\\
 {\sf B}=\lim_{\eps\to 0}V_\eps(\beta_\eps)\
 &=\  \frac{1}{2}H(Z^0_p)+V^0 \,,\\
\nonumber
 {\sf C}:=\lim_{\eps\to 0}\dot V_\eps(\beta_\eps)\
&= \dot{V}^0+
\frac{\dot{U}^0}{8} \bigg(H_{,2}(Z_r^0)^2 + H_{,3}(Z_r^0)^2 +\sigma H_{,4}(Z_r^0)^2 \\ \nonumber
& \qquad\qquad\qquad\qquad\quad +
\frac{1}{\sigma a^2} H(Z_r^0)^2 - \frac{1}{\sigma a^2}\left(\delta^{pq}Z_p^0
H_{,q}(Z_r^0)\right)^2\bigg)\\\nonumber
&\qquad -\frac{\dot{U}^0}{2 \sigma
a^2}\left(\delta^{pq}Z_{p}^0H_{,q}(Z_r^0) -
H(Z_r^0)\right)V^0 + \frac{1}{2}\delta^{pq}H_{,p}(Z_r^0)\dot{Z}_q^0\,. \\\nonumber
\end{align}
\end{proposition}
The calculation is rather technical and we sketch the main
points in Appendix \ref{B}.
Finally we remark that our results are fully
compatible with the ones in \cite{PO:01}:

\begin{remark}
To simplify the comparison of the results on the `limiting geodesics' of
Proposition \ref{prop:explicit_limits} with the heuristically derived geodesics
of the impulsive wave spacetime of \cite[eqs.\ (38),(39)]{PO:01} we remark that
there the geodesics were restricted to $V^0=\dot{Z}^0=0$, and $\dot{U}^0=1$.
Moreover recalling that $1/(\sigma a^2)= \Lambda/3$ and using the notations $G(0)=G(Z_p^0)=
\delta^{pq}Z_p^0 H_{,q}(Z_r^0)-H(Z_r^0)$ and $H(0)=H(Z^0_p)$ of \cite{PO:01},
equations \eqref{eq:ro:consistency} take the form
 \begin{align}\label{eq:ro:final}\nonumber
 {\sf A_i}&=\ \frac{1}{2}\bigg(H_{,i}(0)-\frac{\Lambda}{3}Z^0_i\,G(0)\bigg),\quad
 {\sf A_4}=\frac{1}{2}\bigg(\sigma H_{,4}(0)-\frac{\Lambda}{3}Z^0_4\,G(0)\bigg),
 \quad {\sf B}=\frac{1}{2}H(0),\\
 {\sf C}&=\ \frac{1}{8} \bigg( H_{,2}(0)^2 + H_{,3}(0)^2 + \sigma H_{,4}(0)^2 + \frac{\Lambda}{3}H(0)^2
  - \frac{\Lambda}{3}\left(
  \delta^{pq}Z_p^0 H_{,q}(0)\right)^2 \bigg)\ +\ \dot{V}^0.
 \end{align}
 Finally taking into account that $\epsilon=\mathrm{sign}(\Lambda)$ in \cite{PO:01} as well as
 the slightly different definition of ${\sf C}$ in \cite[eqs.\ (37)]{PO:01} we see that
 \eqref{eq:ro:final} indeed agrees with eqs.\ (38) of \cite{PO:01}.
\end{remark}

\section*{Summary}

In this paper we have rigorously investigated all geodesics in the entire class
of nonexpanding impulsive
gravitational waves propagating in (anti-)de Sitter universe, extending thus
many previous studies of geodesic motion in the class of impulsive
\emph{pp}-wave spacetimes with vanishing cosmological constant. Following
\cite{PO:01} we employed the \emph{distributional form of the metric in the
context of a $5$-dimensional embedding formalism}. We have applied a
regularisation technique, replacing the Dirac-$\delta$ by a general class of
smooth functions---the model delta nets \eqref{eq:delta}. Since we have never used any special property
of the regularising net, our results are completely \emph{regularisation
independent} within this class. In physical terms this means that the formal distributional form of the
impulsive metric \eqref{classical}, \eqref{const} is understood as a limit of a
family of spacetimes with ever shorter but ever stronger sandwich gravitational
waves \eqref{5ipp}, \eqref{const2} of an \emph{arbitrary smooth profile
$\delta_\eps$}.

Although the resulting regularized geodesic equations \eqref{eq:geos} form a
highly complicated coupled system, we were able to prove in Section \ref{sec3}
the existence and uniqueness of geodesics crossing the wave impulse, leading to
completeness results, see Theorems~\ref{thm:ex},~\ref{thm:global} and
\ref{thm:compl_plus},~\ref{thm:compl_minus}. Observe that, in particular, we prove that the geodesics of the
regularised spacetime hitting the wave zone and hence interacting
\emph{nonlinearly} with the impulse actually cross it. This is a physical
information not provided by the approach of \cite{PO:01} in which this feature
is built into the heuristic ansatz.

Our proof is based on the application of a fixed point theorem, the technical details of which can be
found in Appendix~\ref{A}. There we have extended the
range of applicability of this kind of fixed point techniques, originating in
\cite{KS:99a} and generalised in \cite{SS12}, to a far more involved situation
where the higher order `contraction estimate' contains a term
proportional to $1/\eps$, cf.\ \eqref{eq:propA5}. In this way we have
pushed on the crucial point of the method to the estimate \eqref{eq:iint3}, cf.\
the remark in the middle of page \pageref{p:24} below the proof of Proposition \ref{self}. This raises
hopes that these techniques can be extended to the even wilder `$4$-dimensional'
distributional form of the metric \eqref{confppimp} in the future.

In Section~\ref{sec4} we studied boundedness properties of the
global geodesics in the regularized spacetimes. These technical results,
summarized in Proposition~\ref{prop:bounds}, were essential for performing their
limits in the final Section \ref{sec5}. Due to the complexity of
the system  of geodesic equations \eqref{eq:geos} it seemed advisable to
simplify the `usual arguments' in this limiting procedure. We have done so by
abstracting from the concrete form of the (limiting) geodesics and repeatedly
using continuous dependence of solutions of ODEs on its data.
In this way we were able to show that, as the regularisation
parameter goes to zero, the solutions of the regularised geodesic equation
converge to unique geodesics of the background (anti-)de Sitter spacetime
(Theorem~\ref{prop:lim1}) which have to be matched appropriately across the
impulse. In fact, we rigorously derived the explicit form of these matching
conditions (some overly technical related calculations are contained in Appendix
\ref{B}). The resulting coefficients \eqref{eq:ro:consistency} of
Proposition~\ref{prop:explicit_limits} fully agree with previous results derived
by a heuristic approach in~\cite{PO:01}. Remarkably, the impulsive limit is
completely independent of the specific regularization, i.e., in the limit
${\varepsilon\to0}$ it is \emph{the same for any smooth profile} of the sandwich
gravitational waves.

Finally, as mentioned in the introduction in~\cite{PSSS:15}, we have recently
investigated the complete family of nonexpanding impulsive gravitational waves
propagating in spaces of constant curvature (Minkowski, de Sitter and anti-de
Sitter
universes) employing the (Lipschitz) \emph{continuous form of the
metric}\,~\eqref{conti}. Using Filippov's
solution concept for differential equations with discontinuous right hand side
we proved existence and uniqueness
of continuously differentiable geodesics. In section~4 of~\cite{PSSS:15} we
explicitly derived such geodesics using a $\Con^1$-matching procedure resulting in
specific matching conditions, namely equations~(4.4)--(4.10) of~\cite{PSSS:15}.

A natural question thus arises about the mutual consistency of the two results,
both obtained in a rigorous way but starting from two different forms of the
metric, namely the continuous form of the metric (employed in~\cite{PSSS:15})
and the distributional form of the metric in the context of the $5$-dimensional
embedding formalism (employed here and in \cite{PO:01}). In fact, it was shown
in the recent work~\cite{Kar:15} that the \emph{matching conditions} of
\cite{PSSS:15} and \cite{PO:01} \emph{are fully equivalent} when appropriate
coordinate transformations are applied.
This result confirms that both our approaches are consistent. It follows that
the understanding of geodesics in the complete family of spacetimes with
nonexpanding impulsive gravitational waves and any cosmological constant now
rests on firm mathematical grounds. 

These results set the stage for a sound  mathematical
analysis of the `discontinuous coordinate transformations' between the
continuous and the distributional forms of the metric. Together with the
results of \cite{Kar:15} is seems now feasible to rigorously relate the 
continuous form of the metric \eqref{conti} to the `$5$-dimensional'
distributional form \eqref{classical}, \eqref{const}. On the other hand the
technical advances of the fixed point techniques made here might eventually
bring into reach a direct approach on the mathematical intricacies of the
transformation \eqref{trans}.

\section*{Acknowledgement}
We thank Robert \v{S}varc and Milena Stojkovi\'c for taking part in our
discussions in the early stages of this project. C.S. and R.S. were supported by
projects P23714 and P25326 of the Austrian Science Fund (FWF). A.L. was
supported by a 2013 Uni:doc grant of the University of Vienna. J.P. was
supported by the research grant GA\v{C}R~P203/12/0118.

\begin{appendix}
\section{The fixed point argument}\label{A}
In this appendix we detail the fixed point argument used to prove a suitable
existence and uniqueness result for solutions of the regularised
geodesic equations \eqref{eq:geos} with data \eqref{eq:dataeps} that additionally
guarantees the solutions to live long enough to leave the regularisation sandwich.

To do so, we only have to prove existence of the $U_\eps$ and
$Z_{p\eps}$-components, since the equation for $V_\eps$ decouples and is linear,
hence can then be solved on the domain of existence of ($U_\eps$, $Z_{p\eps}$).
Moreover, the sign-difference between the $Z_{i\eps}$-equations and the $Z_{4\eps}$-equation
can safely be ignored in the estimates leading to the fixed point argument.
Therefore, in this appendix, we only (have to) deal with the following
\emph{simplified model system}
\begin{align}\label{eq:math}
\ddot{u}_\eps
&=-\Big( e + \frac{1}{2}\,\dot{u}_\eps^2\,\tilde{G_\eps}
           - \dot{u}_\eps\,\big(H
           \,\delta_{\eps}
           \,u_\eps\dot{\big)}\Big)\
           \frac{u_\eps}{\sigma a^2-u_\eps^2H\delta_\eps}\,, \nonumber\\
 \ddot{z}_\eps-\frac{1}{2}DH\,\delta_{\eps}\dot{u}_\eps^2
&=-\Big(e + \frac{1}{2}\,\dot{u}_\eps^2\,\tilde{G}_\eps
          - \dot{u}_\eps\,
            \big( H\, \delta_\eps\, u_\eps \dot{\big)}\Big)\
   \frac{z_\eps}{\sigma a^2-u_\eps^2H\delta_\eps}\,,
\end{align}
where
$H=H(z_\eps)$ is a smooth function on $\R^3$, $DH$ denotes its gradient, and
$\tilde{G}_\eps(u_\eps,z_\eps):=DH(z_\eps)\,\delta_{\eps}(u_\eps)\,z_\eps +
H(z_\eps)\,\delta'_\eps(u_\eps)\,u_\eps$. We will also frequently use the notation
$x_\eps=(u_\eps,z_\eps)$.
\medskip

We begin by setting up the \emph{initial data}. Let $\eta>0$ and let
$J_\eps=[\alpha_\eps,\alpha_\eps+\eta]$ be the parameter interval where we look
for solutions. In accordance with the strategy employed in Section \ref{sec3},
we will pose initial data at $t=\alpha_\eps$ and compare it to fixed data
(corresponding to the initial data of the `seed geodesic' at $t=0$). So let
 \begin{align}\label{eq:zdata_eps}\nonumber
 &x_\eps^0 =(u_\eps^0 ,z_\eps^0)\in \R\times\R^3\quad \text{and}\quad
 \dot{x}_\eps^0:=(\dot{u}_\eps^0,\dot{z}_\eps^0)\in\R\times\R^3\
 \text{be given and set}\\
 &x_\eps(\alpha_\eps)=(u_\eps(\alpha_\eps),z_\eps(\alpha_\eps))=(u_\eps^0 ,z_\eps^0)
 \quad \text{and}\quad
 \dot x_\eps(\alpha_\eps)=(\dot u_\eps(\alpha_\eps),\dot z_\eps(\alpha_\eps))=(\dot u_\eps^0 ,\dot z_\eps^0)
\end{align}
and let additionally
\begin{equation}\label{eq:zdata}
 \text{$u^0,\dot{u}^0\in\R$, $z^0, \dot{z}^0\in\R^3$ be given and write $x^0 := (u^0,z^0)$,
$\dot{x}^0:=(\dot{u}^0,\dot{z}^0)$}\,.
\end{equation}

As detailed in Section \ref{sec3} we exclusively deal with data satisfying
\begin{align}\label{eq:udata}
  &u_\eps^0=-\eps,\ \dot{u}_\eps^0>0 \ \text{and}\ u^0=0,\ \dot{u}^0>0\,,\
  \text{with the additional assumption (cf.\ \eqref{eq:dataconv})}
 \nonumber\\
  &\text{$x^0_\eps \to x^0$ and  $\dot{x}^0_\eps \to \dot{x}^0$ as
   $\eps\to 0$.}
\end{align}

We will apply our fixed point argument on a \emph{complete metric space}
which we will call the `solution space' and which is given as the closed subset
of $\mathcal{C}^1(J_\eps,\R^4)$ defined by
\begin{align}\label{space}\nonumber
 \mathfrak{X}_\eps := \Big\{ x_\eps=(u_\eps,z_\eps)\in \mathcal{C}^1(J_\eps,\R^4):\
 &x_\eps(\alpha_\eps)=x_\eps^0,\ \dot{x}_\eps(\alpha_\eps)=\dot{x}_\eps^0\ \text{and}\\
 &\|x_\eps-x^0\|_\infty\leq C_1,\ \|\dot{z}_\eps-\dot{z}^0\|_\infty\leq C_2,\
 \dot{u}_\eps\in\Big[\frac{1}{2}\dot{u}^0,\frac{3}{2}\dot{u}^0\Big]\Big\}\,.
\end{align}
Observe that we have `centred' the functions in $\mathfrak{X}_\eps$ around the
`fixed' initial data \eqref{eq:zdata}, while the prospective solutions are required
to assume the $\eps$-dependent data \eqref{eq:zdata_eps} at $t=\alpha_\eps$.
Also note that the final condition forces $\dot u_\eps$ to stay positive, which
is the essential ingredient that forces the solutions to leave the regularisation sandwich.
Now we arrange the constants as follows: First
let $C_1>0$ and set
\begin{align}\label{cond:C2}\nonumber
 C_2:= 1 + \max & \bigg\{ 9\dot{u}^0\|DH\|_\infty\|\rho\|_1,
 \left.\ \, \frac{36}{a^2} \dot{u}^0 (|z^0|+C_1)
\Big(3\|DH\|_\infty\|\rho\|_\infty(|z^0|+C_1)+\|H\|_\infty\|\rho'\|_\infty\Big), \right.\\
&\quad \frac{48}{a^2}(|z^0|+C_1)\left(1 + \frac{3}{2}\dot u^0\|H\|_\infty\Big(\|\rho'\|_\infty
+ \|\rho\|_\infty\Big)\right)\,
\bigg\},
\end{align}
where $\|H\|_\infty$ and $\|DH\|_\infty$ are
taken over the closed Euclidean ball $B_{C_1}(z^0)$. Also $\rho$ is as in \eqref{eq:delta}.
Observe that the space $\mathfrak{X}_\eps$ only depends on $\eps$ via the domain $J_\eps$ and the initial data
\eqref{eq:zdata_eps}.

Next we define the \emph{solution operator} $A_\eps$ acting on $\mathfrak{X}_\eps$ for all $t\in J_\eps$ via
$A_{\veps}(x_\eps)(t):=(A^1_{\veps}(x_\eps)(t)$, $A^2_{\veps}(x_\eps)(t))$ with
\begin{align}\label{eq:def_A}
A^1_\eps(x_\eps)(t)
 &:=-\int_{\alpha_\eps}^t\int_{\alpha_\eps}^s
  \frac{eu_\eps+\tfrac{1}{2}u_\eps\dot{u}_\eps^2\tilde{G}_\eps
          -u_\eps\dot{u}_\eps\left(H\delta_{\veps}u_\eps\right)\dot{}}
       {\sigma a^2-u_\eps^2H\delta_{\veps}}\,\dd r \,\dd s\
  + \dot u_\eps^0(t-\alpha_\eps)-\eps\,,\\
A^2_\eps(x_\eps)(t)
 &:=\int_{\alpha_\eps}^t\int_{\alpha_\eps}^s \left(\,
  \frac{1}{2}DH\delta_\veps\dot{u}_\eps^2
  -\frac{ez_\eps +\tfrac{1}{2}z_\eps\dot{u}_\eps^2\tilde{G}_\eps
              -z_\eps\dot{u}_\eps\left(H\delta_{\veps} u_\eps\right)\dot { }}
   {\sigma a^2-u_\eps^2H\delta_{\veps}}\,\right)\,\dd r \,\dd s
 +\dot{z}_\eps^0(t-\alpha_\eps)+z_\eps^0\,,\nonumber
\end{align}
where again we have suppressed the dependence of $\delta_\eps$, $\tilde G_\eps$
and $H$ as well as their derivatives on the variables.

Our first step will be to show that the operator $A_\eps$ takes $\mathfrak{X}_\eps$ to
$\mathfrak{X}_\eps$, see Proposition \ref{self}
below. We begin with two preliminary results. First we bound the term in the
denominator of $A_\eps$ from below.

\begin{lemma}\label{lem:denom}
Suppose $z\in B_{C_1}(z^0)$ then for all $u\in\R$
\begin{equation}
\left|\ \frac{1}{\sigma a^2-u^2\,H(z)\,\delta_{\eps}(u)}\ \right|\ \leq\
\frac{2}{a^2}\,,
\end{equation}
provided $\veps\leq a^2 / (2 \|\rho\|_\infty \|H\|_\infty)$.

 \end{lemma}
\begin{proof}
First, in case $|u|>\veps$ we have $u\notin \supp(\delta_\veps)$ and
consequently $\displaystyle \frac{1}{|\sigma a^2-u^2H\delta_{\veps}|}=
\frac{1}{a^2}$.
Second, in case  $|u|\leq\veps$ we have $|u^2H(z)\delta_\eps(u)|
\leq \veps \|H\|_\infty\|\rho\|_\infty\leq a^2/2$
and therefore in both cases
\begin{align*}
\frac{1}{|\sigma a^2-u^2H\delta_{\veps}|}\leq \frac{2}{a^2}\,.
\end{align*}
\end{proof}

The second preliminary result shows that the conditions on $\dot u_\eps$ imposed in
\eqref{space}, i.e., that $\dot u_\eps\geq \dot u^0/2$, prevents
the $u$-component from slowing down too much in the sense that $u_\eps(t)$ leaves the
sandwich region early enough. To state the result in a precise way we define
for $x_\eps=(u_\eps,z_\eps)\in \mathfrak{X}_\eps$ the set
\begin{align}
\Gamma_{\veps}(x_\eps)\equiv\Gamma_{\veps}(u_\eps)
:= \{t\in J_\eps\::\:|u_\eps(t)|\leq \veps \}\subseteq J_\eps\,,
\end{align}
which is the maximal set where the terms in \eqref{eq:math}, \eqref{eq:def_A}
involving $\delta_{\veps}$ or $\delta_{\veps}'$ are
non vanishing. We now have:

\begin{lemma}\label{lem:gamma}
The diameter of $\Gamma_{\veps}(u_\eps)$ is bounded for all $x_\eps=(u_\eps,z_\eps)\in \mathfrak{X}_\eps$ by
\begin{align}\label{eq:diam_gamma}
 \diam \left(\Gamma_{\veps}(u_\eps)\right) \leq \frac{4\veps}{\dot{u}^0}\,.
\end{align}
\end{lemma}

\begin{proof}
For $x_\eps\in \mathfrak{X}_\eps$ let $t\in \Gamma_{\veps}(x_\eps)$ which implies
$|u_\eps(t)|\leq\veps$ and so
\begin{equation}
 \veps\geq u_\eps(t) = u_\eps^0 + \int_{\alpha_\eps}^t \dot{u}_\eps(\tau)\,\dd \tau \geq u_\eps^0
 + \frac{1}{2}\dot{u}^0(t-\alpha_\eps)\,.
\end{equation}
But this implies $t\leq \alpha_\eps+2\,(\veps-u_\eps^0)/\dot{u}^0=\alpha_\eps+4\eps/\dot u^0$.
\end{proof}

Now we may state and prove that $A_\eps(\mathfrak{X}_\eps)\subseteq \mathfrak{X}_\eps$ provided $\eta$
is chosen appropriately and $\eps$ is small enough.
\begin{proposition}\label{self}
Set
\begin{align}\label{eq:b}\nonumber
\eta:=\min & \bigg\{
1,\,
\frac{a^2}{24\dot u^0}\,,
\frac{C_1}{\frac{3}{2}+\dot{u}^0}\,,
\frac{2C_1}{54\|\rho\|_1\|DH\|_\infty \dot u^0}\,,
\frac{a^2\,C_1}{12(|z^0|+C_1)}\,,
\frac{a^2\,C_2}{8(|z^0|+C_1)}\,,
\\
&\qquad \frac{C_1 a^2}{54}\Big(\dot u^0(|z^0|+C_1)(3\|DH\|_\infty
\|\rho\|_\infty (|z^0|+C_1) +
\|H\|_\infty \|\rho'\|_\infty)\Big)^{-1},
\frac{C_1}{6(1+|\dot z^0|)},
\nonumber \\
%
&\qquad \frac{C_1a^2}{72}\left(\big(|z^0| +
C_1\big)\Big(3\|DH\|_\infty\|\rho\|_\infty (|\dot{z}^0|+C_2)
   + \frac{3}{2}\dot u^0\|H\|_\infty\big(\|\rho'\|_\infty + \|\rho\|_\infty\big)\Big)\right)^{-1}
   \bigg\}
\end{align}
and
\begin{align}\label{cond:eps00}
 \eps_0' := \min &\bigg\{
\frac{a^2}{2\|\rho\|_\infty\|H\|_\infty}\,,\frac{a^2}{72\dot{u}^0}\, \Big(
3\|DH\|_\infty\|\rho\|_\infty(|z^0|+C_1) +
\|H\|_\infty\|\rho'\|_\infty\Big)^{-1}\,,\nonumber\\
&\qquad\frac{a^2}{96}\,\left(3\|DH\|_\infty\|\rho\|_\infty (|\dot{z}^0|+C_2)
   + \frac{3}{2}\dot u^0\|H\|_\infty\Big(\|\rho'\|_\infty + \|\rho\|_\infty\Big)\right)^{-1}\,,\\
&\qquad\big(3\|DH\|_\infty\|\rho\|_\infty(|\dot{z}^0|+C_2)\big)^{-1}\,,
\frac{\eta\dot { u }^0}{6}\,,\eta\bigg\}\,.\nonumber
\end{align}
Now choose $\eps_0$ such that
\begin{align}\label{cond:eps0}
 &0<\eps_0\leq\eps_0'\,,\quad\text{and}\nonumber\\
 &|\dot{u}_\eps^0 -
\dot{u}^0|\leq \frac{1}{8}\,,\quad |z_\eps^0-z^0|\leq \frac{C_1}{6}\,,\quad
\text{and}\quad |\dot{z}_\eps^0 -
\dot{z}^0|\leq 1\quad \text{for all}\quad 0<\eps\leq\eps_0\,.
\end{align}
Then for all $\eps\leq\eps_0$ the operator $A_{\veps}$ maps $\mathfrak{X}_\eps$
to $\mathfrak{X}_\eps$.
\end{proposition}

Observe that by \eqref{eq:udata} there exists $\eps_0$, that guarantees the
estimates in \eqref{cond:eps0} to hold.

\begin{proof}
We begin by estimating
\begin{equation}\label{A1}
 \frac{d}{dt}A^1_{\veps}(x_\eps)(t)
=-\int_{\alpha_\eps}^t
\frac{eu_\eps+\frac{1}{2}u_\eps\dot{u}_\eps^2\tilde{G}_\eps-u_\eps\dot{u}_\eps\left(H\delta_{\veps}
u_\eps\right)\dot{}}{ \sigma a^2-u_\eps^2H\delta_{\veps}} \,\dd s+\dot u_\eps^0
\end{equation}
and proceed term by term beginning with the latter two under the integral which we will see to
vanish as $\eps\to 0$. Indeed we have for all $\eps\leq\eps_0$
by the definition of $\mathfrak{X}_\eps$ \eqref{space} and by Lemma \ref{lem:denom}
\begin{align}\label{eq:geps}
&\bigg| \int_{\alpha_\eps}^t \frac{u_\eps\dot{u}_\eps^2
           \big(DH \delta_{\eps}z_\eps
           + H\delta'_\eps u_\eps\big)}
           {2\left(\sigma
                   a^2-u_\eps^2H\delta_{\veps}\right)}\dd s \bigg| \nonumber
\\
& \qquad \leq \frac{1}{a^2}\, \diam \left(\Gamma_{\veps}(u_\eps)\right)\,
  \eps\, \Big(\frac{3}{2}\,\dot u^0\Big)^2
\left(3\|DH\|_\infty\|\rho\|_\infty\,\frac{1}{\eps}(|z^0|+C_1)
+ \|H\|_\infty\frac{1}{\eps^2}\|\rho'\|_\infty\,\eps
  \right) \\
&\qquad \leq \frac{9\dot{u}^0 }{a^2}\,\eps\,
  \left( 3\|DH\|_\infty\|\rho\|_\infty(|z^0|+C_1) + \|H\|_\infty\|\rho'\|_\infty
  \right) \leq  \nonumber
\frac{1}{8}\,,
\end{align}
where for the second inequality we have used Lemma
\ref{lem:gamma} and for the third that
\begin{equation}
 \eps_0\leq \frac{a^2}{72\dot{u}^0} \left
(3\|DH\|_\infty\|\rho\|_\infty(|z^0|+C_1) +
\|H\|_\infty\|\rho'\|_\infty\right)^{-1}\,.
\end{equation}

Similarly, we have
\begin{align}\nonumber\label{eq:bracket}
\bigg| \int_{\alpha_\eps}^t &\frac{u_\eps\dot{u}_\eps\left(H\delta_{\veps}u_\eps\right)\dot{}}{\sigma
a^2-u_\eps^2H\delta_{\veps}}\dd s \bigg|&\\
&\leq \frac{12}{a^2}\,\eps\, \left(3\eps \|DH\|_\infty\|\rho\|_\infty (|\dot{z}^0|+C_2)
   + \frac{3}{2}\dot u^0\|H\|_\infty\Big(\|\rho'\|_\infty + \|\rho\|_\infty\Big) \right)
   \leq \frac{1}{8}\,,
\end{align}
where the final estimate again follows from our assumptions on $\eps_0$.

Finally to estimate the first term under the integral in \eqref{A1} we write for $u_\eps\in \mathfrak{X}_\eps$
\begin{align}
 u_\eps(t)=u_\eps^0+\int\limits_{\alpha_\eps}^t\dot u_\eps(s)\,\d s\leq -\eps+\eta\,\frac{3}{2}\dot u^0
 \leq \frac{3}{2}\dot u^0 \eta\,.
\end{align}
Since $-\varepsilon\leq u_\eps(t)$ and by the last condition on $\eps_0$ in
\eqref{cond:eps0} we obtain $|u_\eps(t)|\leq \frac{3}{2}\dot u^0\eta$ and hence
\begin{align}
 \bigg|\int_{\alpha_\eps}^t \frac{eu_\eps}{\sigma a^2-u_\eps^2H\delta_{\veps}}\, \dd s \bigg|
 &\leq \frac{2}{a^2}\int_{\alpha_\eps}^t |u_\eps(s)|\ \dd s
 \leq \frac{3}{a^2}\dot u^0 \eta^2  \leq \frac{3}{a^2}\dot u^0 \eta \leq\frac{1}{8}\,,
\end{align}
where we have used that $\eta\leq 1$ and $\eta\leq a^2/(24\dot u^0)$, cf.\
\eqref{eq:b}.
Thus, by $|\dot{u}_\eps^0 - \dot{u}^0|\leq \frac{1}{8}$ we obtain overall
$\|\frac{\dd}{\dd\,t}A^1_{\veps}(x_\eps)-\dot{u}^0\|_\infty\leq \frac{1}{2}$,
i.e., $\frac{\dd}{\dd\,t}A^1_{\veps}(x_\eps)(t)\in [\frac{1}{2}\dot{u}^0,\frac{3}{2}\dot{u}^0]$
for all $t\in J_\eps$.
\medskip

Moreover, using the above estimates, integrating once more and using
$\eps\leq\eta$ we find that
\begin{equation}
\|A_{\veps}^1(x_\eps)-u^0\|_\infty \leq \eps + \eta \frac{3}{8} +  \eta \dot{u}_\eps^0\leq \eta (\frac{3}{2} +  \dot{u}^0)
\leq C_1\,,
\end{equation}
due to the assumption that $\eta\leq C_1/(\frac{3}{2}+\dot{u}^0)$.
\bigskip

Now we turn to the `spatial component' $A^2_\eps$ of the solution operator. We have to show that
\begin{align}\label{eq:A2}
 \|A_{\eps}^2(x_\eps)&-z^0\|_\infty
\nonumber\\
  &=
   \bigg\| \int_{\alpha_\eps}^t\int_{\alpha_\eps}^s \Big(
   \frac{1}{2}DH\delta_\veps\dot{u}_\eps^2
   -\frac{ez_\eps +\tfrac{1}{2}z_\eps\dot{u}_\eps^2\tilde{G}_\eps
               -z_\eps\dot{u}_\eps\big(H\delta_{\veps} u_\eps\big)\dot{ }}
    {\sigma a^2-u_\eps^2H\delta_{\veps}}\,\Big)\,\dd r \,\dd s\
   +\dot{z}_\eps^0(t-\alpha_\eps)\bigg\|_\infty
\\\nonumber & \leq C_1
\end{align}
and again proceed term by term. To begin with we note the following
auxiliary estimate
\begin{align}\label{eq:3.25}
 \int_{\alpha_\eps}^t|\delta_\eps(u_\eps(s))|\dd s &=
\frac{2}{\dot{u}^0}\int_{\alpha_\eps}^t|\delta_\eps(u_\eps(s))|\frac{\dot{u}^0}{2}\dd s \\
 &\leq \frac{2}{\dot{u}^0}\int_{\alpha_\eps}^t|\delta_\eps(u_\eps(s))| \dot{u}_\eps(s)\dd s =
 \frac{2}{\dot{u}^0}\int_{-\eps}^{u_\eps(t)}|\delta_\eps(r)|\dd r \leq
\frac{2}{\dot{u}^0}\|\rho\|_{1}\,. \nonumber
\end{align}
Now we have once more using the definition of $\mathfrak{X}_\eps$
\begin{align}\label{eq:A21}
\frac{1}{2}\left|\int_{\alpha_\eps}^t \int_{\alpha_\eps}^s DH\delta_{\veps}\dot{u}_\eps^2\, \dd r\dd s \right|
 &\leq \frac{2}{\dot{u}^0}\|\rho\|_1\, \|DH\|_\infty\,\Big(\frac{3}{2}\dot{u}^0\Big)^2 \eta
 = \frac{9}{2}\, \|\rho\|_1\ \|DH\|_\infty \dot u^0\,\eta\,
 \leq \frac{C_1}{6}\,,
 \end{align}
where we have made use of $\eta\leq C_1/((2/54)\|\rho\|_1\|DH\|_\infty \dot u^0)$.
Similarly since $\eta\leq (a^2\,C_1)/(12(|z^0|+C_1))$ we obtain
\begin{align}\label{eq:A22}
&\left|\int_{\alpha_\eps}^t \int_{\alpha_\eps}^s \frac{ez_\eps}{\sigma a^2-u_\eps^2
 H\delta_{\veps}}\,\dd r\,\dd s \right|
\leq \frac{2}{a^2}\left(|z^0|+C_1\right)\,\eta^2
\leq \frac{2}{a^2}\left(|z^0|+C_1\right)\,\eta
\leq \frac{C_1}{6}\,.
\end{align}
Furthermore, we estimate
\begin{align*}
&\left|\frac{1}{2}\int_{\alpha_\eps}^t \int_{\alpha_\eps}^s \frac{z_\eps\dot{u}_\eps^2\tilde{G}_\eps}{\sigma
a^2-u_\eps^2H\delta_{\veps}}\,\dd r\,\dd s\right|\\
&\qquad\leq \frac{9}{a^2}\dot{u}^0 \left(|z^0| + C_1\right)
\left(3\|DH\|_\infty\|\rho\|_\infty (|z^0|+C_1)  +
\|H\|_\infty \|\rho'\|_\infty \right)\, \eta\leq\frac{C_1}{6}\,,
\end{align*}
where we have used $\eta\leq \frac{C_1 a^2}{54}\Big(\dot
u^0(|z^0|+C_1)(3\|DH\|_\infty \|\rho\|_\infty  (|z^0|+C_1) +
\|H\|_\infty \|\rho'\|_\infty)\Big)^{-1}$,
and finally
\begin{align}\nonumber
&\left|\int_{\alpha_\eps}^t \int_{\alpha_\eps}^s \frac{z_\eps\dot{u}_\eps\left(H\delta_{\veps}u_\eps\right)\dot{}}{\sigma
a^2-u_\eps^2H\delta_{\veps}}\,\dd
r\,\dd s\right|\\
&\qquad\leq \frac{12}{a^2} \left(|z^0| + C_1\right)
\left(3\eps\|DH\|_\infty\|\rho\|_\infty (|\dot{z}^0|+C_2)
   + \frac{3}{2}\dot u^0\|H\|_\infty\Big(\|\rho'\|_\infty + \|\rho\|_\infty\Big)\right)\,
\eta \leq \frac{C_1}{6}\,, 
\end{align}
where we have used the final condition on $\eta$ in \eqref{eq:b}. This establishes \eqref{eq:A2} using the
one before last condition on $\eta$ in \eqref{eq:b} together with $|z_\eps^0-z^0|\leq C_1/6$.
\bigskip

It remains to show $\|\frac{\dd}{\dd\,t}A^2_{\veps}(x_\eps) - \dot{z}^0\|_\infty\leq C_2$. As in \eqref{eq:A21},
\eqref{eq:A22} we estimate
\begin{align}\label{eq:iint1}
 \frac{1}{2}\left|\int_{\alpha_\eps}^t DH\delta_{\veps}(u_\eps)\dot{u}_\eps^2\,\dd s \right|
 &\leq \frac{9}{4}\, \|\rho\|_1\ \|DH\|_\infty \dot u^0 \leq \frac{C_2}{4}\,,
\\\nonumber
 \left|\int_{\alpha_\eps}^t \frac{ez_\eps}{\sigma a^2-u_\eps^2H\delta_{\veps}}\,\dd s \right|
 &\leq \frac{2}{a^2}\left(|z^0|+C_1\right) \eta \leq \frac{C_2}{4}\,,
\end{align}
where we have used the first condition on $C_2$ in \eqref{cond:C2} and
the sixth one on $\eta$ in \eqref{eq:b}.

For the remaining two terms we have
\begin{align}\label{eq:iint2}\nonumber
&\left|\frac{1}{2}\, \int_{\alpha_\eps}^t \frac{z_\eps\dot{u}_\eps^2\tilde{G}_\eps}{\sigma a^2-u_\eps^2H\delta_{\veps}}\,\dd
s\right|\\
&\quad \leq \frac{36}{4 a^2}\dot{u}^0\left(|z^0| + C_1\right) \left(3\|DH\|_\infty \|\rho\|_\infty (|z^0|+C_1) +
\|H\|_\infty \|\rho'\|_\infty \right)\, \leq\frac{C_2}{4}\,,
\end{align}
where we have used the second condition on $C_2$ in \eqref{cond:C2}, and
\begin{align}\label{eq:iint3} \nonumber
&\left|\int_{\alpha_\eps}^t \frac{z_\eps\dot{u}_\eps\left(H\delta_{\veps}u_\eps\right)\dot{}}{\sigma
a^2-u_\eps^2H\delta_{\veps}}\,\dd s\right|\\
&\qquad \leq \frac{12}{a^2} \left(|z^0| + C_1\right) \left(3\eps \|DH\|_\infty\|\rho\|_\infty (|\dot{z}^0|+C_2) +
\frac{3}{2}\dot u^0\|H\|_\infty\Big(\|\rho'\|_\infty + \|\rho\|_\infty\Big)\right)\\
&\qquad \leq \frac{12}{a^2} \left(|z^0| + C_1\right)
\left(1 + \frac{3}{2}\dot u^0\|H\|_\infty\Big(\|\rho'\|_\infty
+ \|\rho\|_\infty\Big)\right) \leq \frac{C_2}{4}\,. \nonumber
\end{align}
Here we have used the fourth condition on $\eps_0$ in \eqref{cond:eps0} as well as the final
condition on $C_2$ in \eqref{cond:C2}.
\end{proof}

Observe that in the estimate \eqref{eq:iint3} it is absolutely \emph{vital}
\label{p:24} that
the term in the second line involving $C_2$ is proportional to $\eps$
---\,otherwise we would end up in a circle and our method would fail.

Our next step is to prove that the solution operator $A_\eps$ has a fixed point on $\mathfrak{X}_\eps$. To this end we need
the
following technical preparation.
\begin{lemma}\label{prop:est_deltau}
 There exist constants $\tilde C$ and $\tilde C'$ (independent of $\eps$)
 such that for all $x_\eps,x_\eps^*\in \mathfrak{X}_\eps$ we have
 \begin{enumerate}
 \item[(i)\ ] $\displaystyle
   \left|\int_{\alpha_\eps}^t\left(\delta_{\veps}(u_\eps)u_\eps-\delta_{\veps}(u_\eps^*)u_\eps^* \right)\dd s \right| \leq
\tilde{C} \|u_\eps-u_\eps^*\|_\infty\,,$
   and \\ \ \\
 \item[(ii)\ ] $\displaystyle
   \left|\int_{\alpha_\eps}^t\left(\delta'_{\veps}(u_\eps)u_\eps^2-\delta'_{\veps}(u_\eps^*)(u_\eps^*)^2 \right)\dd s \right|
\leq \tilde{C}' \|u_\eps-u_\eps^*\|_\infty\,.$
 \end{enumerate}
\end{lemma}
\begin{proof}
To prove (i) 
we first consider the case $\|u_\eps-u_\eps^*\|_\infty\leq \veps$.
We have by \eqref{eq:3.25}
\begin{align}
 \left|\int_{\alpha_\eps}^t\left(\delta_{\veps}(u_\eps)u_\eps-\delta_{\veps}(u_\eps^*)u_\eps^* \right)\dd s \right|
 &\leq \int_{\alpha_\eps}^t \left|\delta_{\veps}(u_\eps)u_\eps-\delta_{\veps}(u_\eps)u_\eps^*\right|\,\dd s +
   \int_{\alpha_\eps}^t\left|\delta_{\veps}(u_\eps)u_\eps^*-\delta_{\veps}(u_\eps^*)u_\eps^*\right| \,\dd s \nonumber\\
 &\leq\frac{2}{\dot{u}^0}\|\rho\|_{1}\|u_\eps-u_\eps^*\|_\infty +
\int_{\Gamma_\veps(x_\eps)\cup\Gamma_{\veps}(x_\eps^*)}
\left|\delta_{\veps}(u_\eps)-\delta_{\veps}(u_\eps^*)\right|\left|u_\eps^*\right| \,\dd s\,.
\end{align}
Now the last integral is non-vanishing only if $|u_\eps|\leq \veps$ or $|u_\eps^*|\leq \veps$ hence we have in any case
$|u_\eps^*|\leq 2\veps$. Since both $x_\eps$, $x_\eps^*\in \mathfrak{X}_\eps$, Lemma \ref{lem:gamma} applies so that
using a mean value argument we may bound the integral by
$\frac{8}{\dot{u}^0}\|\rho'\|_\infty\|u_\eps-u_\eps^*\|_\infty$.

In case $\|u_\eps-u_\eps^*\|_\infty>\veps$ we obtain again from Lemma \ref{lem:gamma}
\begin{align}\nonumber
 \left|\int_{\alpha_\eps}^t\left(\delta_{\veps}(u_\eps)u_\eps-\delta_{\veps}(u_\eps^*)u_\eps^* \right)\dd s \right|
&\leq \int_{\Gamma_\eps(x_\eps)}|\delta_{\veps}(u_\eps)u_\eps|\,\dd s +
\int_{\Gamma_\eps(x_\eps^*)}|\delta_{\veps}(u_\eps^*)u_\eps^*|\,\dd s\\
&\leq \frac{8}{\dot{u}^0}\|\rho\|_\infty\,\eps < \frac{8}{\dot{u}^0}\|\rho\|_\infty\|u_\eps-u_\eps^*\|_\infty\,.
\end{align}
So we may chose $\tilde{C}=\frac{2}{\dot{u}^0}\max\left(\|\rho\|_{1} + 4\|\rho'\|_\infty, 4\|\rho\|_\infty\right)$ and
(i) is proved.
\medskip

(ii) is proved analogously with the choice
$\tilde{C}'= \frac{4}{\dot{u}^0}\max\left(4\|\rho'\|_\infty + \|\rho''\|_\infty\,,2\|\rho'\|_\infty\right)$.
\end{proof}

We finally prove the key estimates which will allow the application
of Weissinger's fixed point theorem.
\begin{proposition}\label{prop:weiss}
 There exists a sequence of positive real numbers $(\alpha_n)_n$ (depending on
 $\rho$, $\rho'$, $\rho''$, $H$, $DH$, $D^2H$, and $\dot{u}^0$ but independent of $\eps$) with
 $\sum_{n\in\N}\alpha_n<\infty$ such that for all $x_\eps,x_\eps^*\in \mathfrak{X}_\eps$ with $\eps\leq\eps_0$ of
 \eqref{cond:eps0} and $\eta$ as in \eqref{eq:b} and all $n\in\N$ we have
\begin{align}\label{eq:propA5}
 \|(A_\eps)^n(x_\eps)- (A_\eps)^n(x_\eps^*)\|_{\Con^1} \leq \frac{1}{\eps}\alpha_n \|x_\eps-x_\eps^*\|_{\Con^1}\,.
\end{align}
\end{proposition}

\begin{proof}
It suffices to show  $\|A_\eps(x_\eps)-A_\eps(x_\eps^*)\|_{\Con^1} \leq
(C/\eps) \|x_\eps-x_\eps^*\|_{\Con^1}$ for some appropriate constant $C$, since for higher
powers we then may use
\begin{equation}\label{eq-A-n}
 \int_{\alpha_\eps}^{t_{2n}}\ldots\int_{\alpha_\eps}^{t_{1}}1\,\dd t\dd t_{1}\ldots\dd t_{2n-1}
\leq \frac{\eta^{2n}}{(2n)!}\,
\end{equation}
to obtain a converging series.

We again proceed term by term, skipping some of the details of the (by
now) routine estimates and only stress the technical key points.

We start with the first term in $\|A_\eps^1(x_\eps)-A_\eps^1(x_\eps^*)\|_{\Con^1}$. By
writing the two summands on a common denominator we obtain
\begin{align} \nonumber
 &\left|\int_{\alpha_\eps}^t \frac{eu_\eps}{\sigma a^2 - u_\eps^2H(z_\eps)\delta_\veps(u_\eps)}  \dd s
    - \frac{eu_\eps^*}{\sigma a^2 - (u_\eps^*)^2H(z_\eps^*)\delta_\veps(u_\eps^*)} \,\dd s\right|
\\
&\qquad\leq \frac{4}{a^4} \int_{\alpha_\eps}^t a^2|u_\eps-u_\eps^*| \d s
 +\frac{4}{a^4}
\int_{\alpha_\eps}^t|u_\eps(u_\eps^*)^2H(z_\eps^*)\delta_\eps(u_\eps^*)-u_\eps(u_\eps^*)^2H(z_\eps)\delta_\eps(u_\eps^*)| \d s
\nonumber \\
 &\qquad\ +\frac{4}{a^4}
   \int_{\alpha_\eps}^t|u_\eps(u_\eps^*)^2H(z_\eps)\delta_\eps(u_\eps^*)-u_\eps^2u_\eps^*H(z_\eps)\delta_\eps(u_\eps)| \d s
\nonumber \\
&\qquad \leq \frac{4}{a^2} \eta \|u_\eps-u_\eps^*\|_\infty +
 \frac{4}{a^4}(|u^0|+C_1)^2\left(\lip(H)\|z_\eps-z_\eps^*\|_\infty\frac{4}{\dot{u}^0}
 \|\rho\|_\infty + \|H\|_\infty
 \tilde{C}\|u_\eps-u_\eps^*\|_\infty\right)\\
&\qquad \leq \frac{4}{a^2}\left(\eta +
\frac{1}{a^2}(1+C_1)^2\Big(\lip(H)\frac{4}{\dot{u}^0}\|\rho\|_\infty +
 \|H\|_\infty \tilde{C}\Big)\right)\, \|x_\eps-x_\eps^*\|_{\Con^1}\,, \nonumber
\end{align}
where $\lip(H)$ denotes the Lipschitz constant of $H$ on $B_{C_1}(z^0)$ and
$\tilde{C}$ is the constant given by Lemma \ref{prop:est_deltau}.

For the second term we need the following auxiliary estimate which is proven by
a combination of (i) and (ii) in Lemma \ref{prop:est_deltau}:
\begin{align}\label{eq:Gu}
 \int_{\alpha_\eps}^t|\tilde{G} u_\eps - \tilde{G}^* u_\eps^*|\dd s \leq C'\|x_\eps-x_\eps^*\|_{\Con^1}\,,
\end{align}
where $C'=\|D^2H\|_\infty \|\rho\|_\infty (|z^0|+C_1) + \|DH\|_\infty
\big((|z^0|+C_1)\tilde{C} + \|\rho\|_\infty + \|\rho'\|_\infty\big) + \|H\|_\infty\tilde{C}'$.

Abbreviating $C_{\tilde{G}}:= 3 \|DH\|_\infty \|\rho\|_\infty(|z^0|+C_1) + \|H\|_\infty\|\rho'\|_\infty$ we are able to
estimate
\begin{align*}
 &\left|\int_{\alpha_\eps}^t \left(\frac{\tfrac{1}{2}\, u_\eps \dot{u}_\eps^2\tilde{G}}{\sigma a^2-u_\eps^2
H(z_\eps)\delta_{\veps}(u_\eps)} -
\frac{\tfrac{1}{2}\, u_\eps^* (\dot{u}_\eps^*)^2\tilde{G}^*}{\sigma a^2-(u_\eps^*)^2
H(z_\eps^*)\delta_{\veps}(u_\eps^*)}\right)\dd s\right|\\
&\leq\frac{2}{a^4}\int_{\alpha_\eps}^t a^2\left|u_\eps\dot{u}_\eps^2 \tilde{G} - u_\eps^*(\dot{u}_\eps^*)^2\tilde{G}^*\right|
+
 \left|u_\eps\dot{u}_\eps^2 \tilde{G}(u_\eps^*)^2H(z_\eps^*)\delta_\eps(u_\eps^*) -
u_\eps^*(\dot{u}_\eps^*)^2\tilde{G}^*u_\eps^2 H(z_\eps)\delta_\eps(u_\eps)\right| \, \dd
s\\
&\leq \frac{18 (\dot{u}^0)}{a^4}\bigg(\frac{a^2 C'}{4} + \frac{a^2 C_{\tilde{G}}}{3\dot{u}^0}
 +\frac{C'}{4} \|H\|_\infty\|\rho\|_\infty + \frac{\tilde{C} C_{\tilde{G}}}{4} \|H\|_\infty +
\frac{C_{\tilde{G}}}{4}\lip(H)\|\rho\|_\infty + \frac{C_{\tilde{G}}}{3\dot{u}^0}
\|H\|_\infty\|\rho\|_\infty \bigg)\\
&\qquad\times \|x_\eps-x_\eps^*\|_{\Con^1}\,,
\end{align*}
by using \eqref{eq:Gu} and Lemma \ref{prop:est_deltau}.
\bigskip

The final term in $A^1_\eps(x_\eps)$, i.e.,
\begin{equation}
\left|\int_{\alpha_\eps}^t \left(\frac{u_\eps\dot{u}_\eps
\Big(H(z_\eps)\delta_\eps(u_\eps)u_\eps\Big)\dot{} }{\sigma a^2 - u_\eps^2 H(z_\eps)\delta_\eps(u_\eps)} -
\frac{u_\eps^*\dot{u}_\eps^*\Big(H(z_\eps^*)\delta_\eps(u_\eps^*)u_\eps^*\Big)\dot{} }
{\sigma a^2 - (u_\eps^*)^2 H(z_\eps^*)\delta_\eps(u_\eps^*)}
\right)\dd s\right|
\end{equation}
can be estimated in perfect analogy to the previous terms inserting and
subtracting appropriate terms wherever necessary to arrive at an estimate
proportional to $\|x_\eps-x_\eps^*\|_{\Con^1}$.
\medskip

The `spatial component' $A^2_\eps$ of the solution operator can be treated in a
similar way. The only new aspect when estimating
$\|A_\eps^2(x_\eps)-A_\eps^2(x_\eps^*)\|_{\Con^1}$ is the following. When bounding
terms like $|\tilde{G}-\tilde{G}^*|$ by multiples of $\|x_\eps-x_\eps^*\|_{\Con^1}$
we find that they are no longer multiplied by $u_\eps$ and $u_\eps^*$, respectively.
Thus we cannot use the auxiliary result \eqref{eq:Gu} and
consequently terms proportional to $1/\eps$ remain.
(Note, however, that the occurrence of $1/\eps$-terms at this stage causes
no problem at all in the application of the fixed point theorem, see below.)
Summing up we arrive at
\begin{equation*}
\|\frac{\dd}{\dd\,t}A_{\veps}(x_\eps) - \frac{\dd}{\dd\,t}A_{\veps}(x_\eps^*)\|_\infty \leq
\frac{1}{\eps}C \|x_\eps-x_\eps^*\|_{\Con^1}\,,
\end{equation*}
where $C$ is some constant (as above depending on $H$, $\rho$, etc.). Furthermore, since $\eta\leq 1$ we obtain the same
estimate for the zeroth order, i.e.,
$ \| A_\eps(x_\eps)- A_\eps(x_\eps^*)\|_\infty \leq \frac{1}{\eps}C \|x_\eps-x_\eps^*\|_{\Con^1}$,
and hence
$$\| A_\eps(x_\eps)- A_\eps(x_\eps^*)\|_{\Con^1} \leq \frac{1}{\eps}C
\|x_\eps-x_\eps^*\|_{\Con^1}.$$

Finally, for higher powers of $A_\eps$ we obtain (using \eqref{eq-A-n})
\begin{equation*}
 \|(A_\eps)^n(x_\eps)- (A_\eps)^n(x_\eps^*)\|_{\Con^1} \leq \frac{1}{\eps}\alpha_n \|x_\eps-x_\eps^*\|_{\Con^1}\,,
\end{equation*}
where $\alpha_n:= C \frac{\eta^{2n}}{(2n)!}$ ($n\in\N$).
\end{proof}

At this point we finally obtain the existence of a unique solution to
\eqref{eq:math} in $\mathfrak{X}_\eps$ for all fixed small $\eps$ by applying
Weissinger's fixed point theorem (\cite{Wei:52}).
Note that the factor $1/\veps$ in the estimate \eqref{eq:propA5} provided by
Proposition \ref{prop:weiss}
does not cause any trouble. Its only effect is that the approximating sequence
$(A_\eps)^n (x_\eps)$ converges to the fixed point slower as $\eps$ gets
smaller. Nevertheless we obtain a fixed point for every fixed (small) $\eps$:

\begin{theorem}[Existence and uniqueness]\label{thm:exmat}
Consider the system \eqref{eq:math} with initial data
\eqref {eq:zdata_eps}, satisfying \eqref{eq:zdata},  \eqref{eq:udata}.
Then for all $\eps\leq\eps_0$ where $\eps_0$ is constrained
by \eqref{cond:eps0} and for $\eta$ given by \eqref{eq:b} we have a unique
smooth solution $(u_\eps,z_\eps)$ on $[\alpha_\eps,\alpha_\eps+\eta]$. Moreover,
$u_\eps$ and $z_\eps$ as well as their first order derivatives are
uniformly bounded in $\eps$.
\end{theorem}

\begin{proof}
Propositions \ref{self} and \ref{prop:weiss} allow the application of
Weissinger's fixed point theorem (\cite{Wei:52}) for fixed $\eps\leq \eps_0$ and
suitable $\eta$, providing thus a unique fixed point for the operator $A_\eps$
on the space $\mathfrak{X}_\eps$ which in
turn gives a unique ${\mathcal C}^1$-solution $x_\eps=(u_\eps,z_\eps)$
on $[\alpha_\eps,\alpha_\eps+\eta]$  to the system \eqref{eq:math} with data \eqref{eq:zdata_eps}.
Moreover, since the right hand sides of \eqref{eq:math} are smooth the solution is smooth as well.

The solution obtained via the fixed point argument is unique in the space $\mathfrak{X}_\eps$
and thereby unique among all smooth solutions assuming this data by the usual
argument from ODE-theory.

Finally,  $u_\eps$, $\dot u_\eps$, $z_\eps$, and $\dot z_\eps$ are bounded
uniformly in $\eps$ on $[\alpha_\eps,\alpha_\eps+\eta]$ by the very definition of $\mathfrak{X}_\eps$.
\end{proof}

\section{Limits}\label{B}

In this appendix we deal with the explicit
form of the limits ${\sf A_p}=\lim_{\eps\to 0} \dot Z_{p\eps}(\beta_\eps)$,
${\sf B}=\lim_{\eps\to 0} V_\eps(\beta_\eps)$, and ${\sf C}=\lim_{\eps\to 0} \dot
V_\eps(\beta_\eps)$ as stated in Proposition \ref{prop:explicit_limits}. Since
the actual calculations are overly technical we only sketch the main points.

Again the sign-difference between the $Z_i$-components and $Z_4$ is minor and
to simplify the notation  we will use a similar convention
as in Appendix \ref{A} and write $Z_\eps$ and $Z$ instead of $Z_{p\eps}$ and
$Z_p$ and analogously for their derivatives. Also we will write $DH$ instead of
$H_{,p}$.

Starting with ${\sf A_p}$ we use the differential equation \eqref{eq:geos}
for ${Z_{p\ep}}$ and the uniform converge of $Z_{p\ep}$ and $\dot{U}_\eps$
established in Theorem~\ref{prop:lim1} to show that
\begin{equation}\label{eq:simple_lim_a}
 {\sf A}=\lim_{\eps\to 0}\dot Z_{\eps}(\beta_\eps)\ =\ \frac{1}{2}\,\dot U^0\,
          \Bigl(DH(Z^0) + \frac{Z^0}{\sigma a^2}\bigl(H(Z^0)
           - DH(Z^0)Z^0 \bigr)\Bigr)+\dot Z^0.
\end{equation}

To begin with we express $\dot{Z}_\ep(\beta_\ep)$ according to \eqref{eq:geos}
\begin{align}
 \dot{Z}_\ep(\beta_\eps) &= \dot{Z}_\eps^0 + \int_{\alpha_\eps}^{\beta_\eps}
\ddot{Z}_\eps(r) \,\d r\nonumber\\
 &= \dot{Z}_\eps^0 + \frac{1}{2}\int_{\alpha_\eps}^{\beta_\eps} DH \delta_\eps
\dot{U}^2 \,\d r -  \int_{\alpha_\eps}^{\beta_\eps} \frac{e Z_\eps}{\sigma a^2 -
U_\eps^2 H(Z_\eps)\delta_\eps} \,\d r\nonumber\\
&\qquad- \frac{1}{2}\int_{\alpha_\eps}^{\beta_\eps} \frac{\dot{U}_\eps^2 DH
\delta_\eps Z_\eps^2}{\sigma a^2 - U_\eps^2 H(Z_\eps)\delta_\eps} \,\d r +
\frac{1}{2}\int_{\alpha_\eps}^{\beta_\eps} \frac{\dot{U}_\eps^2 H \delta'_\eps
U_\eps Z_\eps}{\sigma a^2 - U_\eps^2 H(Z_\eps)\delta_\eps} \,\d r\\
&\qquad+\int_{\alpha_\eps}^{\beta_\eps} \frac{\dot{U}_\eps DH \dot{Z}_\eps
\delta_\eps U_\eps Z_\eps}{\sigma a^2 - U_\eps^2 H(Z_\eps)\delta_\eps} \,\d r +
\int_{\alpha_\eps}^{\beta_\eps} \frac{\dot{U}_\eps^2 H \delta_\eps
Z_\eps}{\sigma a^2 - U_\eps^2 H(Z_\eps)\delta_\eps} \,\d r\nonumber\\
&=:\dot{Z}_\eps^0 + \rom{I}_\eps + \rom{II}_\eps + \rom{III}_\eps + \rom{IV}_\eps +
\rom{V}_\eps + \rom{VI}_\eps\,,\nonumber
\end{align}
where we have used that
\begin{align}\label{eq:terms}
 \frac{1}{2}\dot{U}_\ep^2 \tilde{G}_\ep -
\dot{U}_\ep(H(Z_\ep)\delta_\ep U_\ep)\dot{} = \frac{1}{2}\dot{U}_\ep^2 DH
\delta_\ep Z_\ep - \frac{1}{2}\dot{U}_\ep^2
H\delta_\ep' U_\ep - \dot{U}_\ep DH \dot{Z}_\ep\delta_\ep U_\ep - \dot{U}_\ep^2
H\delta_\ep\,.
\end{align}

Proceeding term by term we have
\begin{align}\nonumber\label{eq:1stterm}
 \Bigl|\rom{I}_\eps - \frac{1}{2}DH(Z^0)\dot{U}^0\Bigr|
 &=\frac{1}{2} \Bigl| \int_{-\ep}^\ep
\left(DH(Z_\ep(U_\ep^{-1}(s)))\delta_\ep(s)\dot{U}_\ep(U_\ep^{-1}
(s)) -
DH(Z^0)\delta_\ep(s)\dot{U}^0\right)\,\dd s \Bigr|\\ \nonumber
&\leq {\frac{1}{2}\sup}_{w\in U_\ep^{-1}([-\ep,\ep]) 
} \Bigr|DH(Z_\ep(w))\dot{U}_\ep(w) -
DH(Z^0)\dot{U}^0\Bigr|\ \|\rho\|_{L^1}\ \to 0\,,
\end{align}
where we have used that $U_\ep^{-1}([-\ep,\ep]) = [\alpha_\eps,\beta_\eps]$
together with Lemma \ref{lem:gamma}.
The next term, $\rom{II}_\eps$, vanishes in the limit by the uniform boundedness of
the integrand, the same holds true for $\rom{V}_\eps$.
Now $\rom{III}_\eps$ can be treated as $\rom{I}_\eps$, additionally using
\eqref{eq:estdenom} to conclude
\begin{align}
  \rom{III}_\eps \to  -\frac{1}{2}\frac{\dot{U}^0 DH(Z^0) (Z^0)^2}{\sigma
a^2}\,.
\end{align}
We treat $\rom{IV}_\eps$ using $\int \delta_\ep'(s) s\,\dd s = -1$
to obtain
\begin{align}\nonumber
 \Bigl|\rom{IV}_\eps + \frac{\dot{U}^0 H(Z^0)Z^0}{2\sigma a^2}\Bigr|\\
&\hspace{-2cm}\leq {\frac{1}{2}\sup}_{w\in
U_\ep^{-1}([-\ep,\ep])}\Bigl|\frac{\dot{U}_\ep(w)H(Z_\ep(w))Z_\ep(w)}{\sigma a^2
-U^2_\ep(w) H(Z_\ep(w))\delta_\ep(U_\ep(w))} - \frac{\dot{U}^0 H(Z^0)Z^0}{\sigma a^2}\Bigr|\
\|\delta_\ep'(s)s\|_{L^1}\ \to 0\,.
\end{align}

Finally, the limit of $\rom{VI}_\eps$  is proportional to the limit of $\rom{IV}_\eps$,
\begin{align}\nonumber
 \Bigl|\rom{VI}_\eps - \frac{\dot{U}^0 H(Z^0)Z^0}{\sigma a^2}\Bigr|\\
&\hspace{-2cm}\leq  \sup_{w\in
U_\ep^{-1}([-\ep,\ep])}\Bigl|\frac{\dot{U}_\ep(w)H(Z_\ep(w))Z_\ep(w)}{\sigma a^2
-
U^2_\ep(w) H(Z_\ep(w))\delta_\ep(U_\ep(w))} - \frac{\dot{U}^0 H(Z^0)Z^0}{\sigma
a^2}\Bigr|\
\|\rho\|_{L^1}\to 0 \,.
\end{align}
By adding up the terms and using \eqref{eq:dataconv} we establish
\eqref{eq:simple_lim_a}.
\medskip

The calculations for ${\sf B}$ are relatively simple.
Using equation \eqref{eq:geos} for $\ddot{V}_\eps$ we
write (cf.\ \eqref{eq:Vepsteps})
\begin{align}
 V_\eps(\beta_\eps) = V^0_\eps + \frac{1}{2}
\int_{\alpha_\eps}^{\beta_\eps}\int_{\alpha_\eps}^s H(Z_\eps(r))\delta'_\eps(U_\eps(r)) \dot{U}_\eps(r)^2\,\dd r\, \dd s +
O(\eps)\,.
\end{align}
We substitute twice, use $\int_{-\eps}^{\eps}\int_{-\eps}^s
\delta_\eps'(r)\,\dd r\, \dd s = 1$ and insert appropriate terms to obtain
\begin{align}\nonumber
 \frac{1}{2}\biggl| \int_{\alpha_\eps}^{\beta_\eps}\int_{\alpha_\eps}^s
& H(Z_\eps(r))\delta'_\eps(U_\eps(r)) \dot{U}_\eps(r)^2\,\dd r\,
\dd s - H(Z^0)\biggr|\\
 = &\frac{1}{2}\biggl|
\int_{-\eps}^{\eps}\frac{1}{\dot{U}_\eps(U_\eps^{-1}(l))}\int_{-\eps}^l
H(Z_\eps(U_\eps^{-1}(\tau)))\delta'_\eps(\tau)
\dot{U}_\eps(U_\eps^{-1}(\tau))\,\dd \tau\,\dd l\nonumber \\
&\hspace*{4cm} -
\int_{-\eps}^{\eps}\frac{\dot{U}_\eps(U_\eps^{-1}(l))}{\dot{U}_\eps(U_\eps^{-1}(l))}\int_{-\eps}^l H(Z^0)
\delta_\eps'(\tau)\,\dd \tau\,\dd l\biggr|\\
\leq &\frac{4 \|\rho'\|_\infty}{\dot{U}^0}\ \biggl( \sup_{w\in
U_\eps^{-1}([-\eps,\eps])} |H(Z_\eps(w))\dot{U}_\eps(w) -
H(Z^0)\dot{U}^0| \nonumber\\
&\hspace*{4cm} + |H(Z^0)| \sup_{w\in U_\eps^{-1}([-\eps,\eps])}
|\dot{U}_\eps(w) -\dot{U}^0|\ \biggr)\,,
\nonumber
\end{align}
which goes to zero by the uniform convergence of $Z_\eps$ and $\dot{U}_\eps$,
establishing the claimed form of ${\sf B}$.
%
%
\medskip

Finally we turn to the calculation of ${\sf C}$ which is the most demanding one.
As above we express $\dot{V}_\eps(\beta_\eps)$ using the geodesic equation
\eqref{eq:geos} to obtain
\begin{align}\label{eq:V_dot}\nonumber
 \dot{V}_{\eps} (\beta_\eps)&= \dot{V}_{\eps}^0 + \int_{\alpha_\eps}^{\beta_\eps}
\frac{1}{2}H(Z_{\eps}(r))\delta_{\eps}'(U_\eps(r))\dot{U}^2_\eps(r)\,\dd r +\int_{\alpha_\eps}^{\beta_\eps}
DH(Z_{\eps}(r))\delta_{\eps}(U_\eps(r))\dot{U}_\eps(r)\dot{Z}_\eps(r)\,\dd r\\
 & - \int_{\alpha_\eps}^{\beta_\eps} \frac{e\left(V_\eps(r) + H(Z_\eps(r))\delta_{\eps}(U_\eps(r))U_\eps(r)\right)}{\sigma
a^2- H(Z_\eps(r))\delta_{\eps}(U_\eps(r))U^2_\eps(r)}\,\dd r \\\nonumber
 & -  \int_{\alpha_\eps}^{\beta_\eps} \frac{\left(\frac{1}{2}\dot{U}^2_\eps(r)\tilde{G}_\eps(r) -
\dot{U}_\eps(r)\big(H(Z_\eps(r))U_\eps(r)\delta_\eps(U_\eps(r))\dot{\big)}\right)V_\eps(r)}{\sigma a^2-
H(Z_\eps(r))\delta_{\eps}(U_\eps(r))U^2_\eps(r)}\,\dd r\, \\\nonumber
 & -  \int_{\alpha_\eps}^{\beta_\eps} \frac{\left(\frac{1}{2}\dot{U}^2_\eps(r)\tilde{G}_\eps(r) -
\dot{U}_\eps(r)\big(H(Z_\eps(r))U_\eps(r)\delta_\eps(U_\eps(r))\dot{\big)}\right)H(Z_\eps(r))\delta_{\eps}
(U_\eps(r))U_\eps(r)}{\sigma a^2-
H(Z_\eps(r))\delta_{\eps}(U_\eps(r))U^2_\eps(r)}\,\dd r\\
&=: \dot{V}_\eps^0 + \rom{I}_\eps + \rom{II}_\eps + \rom{III}_\eps + \rom{IV}_\eps + \rom{V}_\eps\nonumber\,.
\end{align}
Note that $\rom{III}_\eps\to 0$ because the integrand is uniformly bounded. Now
we rewrite $\rom{I}_\eps$ substituting $s=U_\eps(r)$, abbreviating
$w:=U_\eps^{-1}(s)$, and using equation \eqref{eq:geos} for $\ddot{U}_\eps$
\begin{align}\label{eq:I_eps}\nonumber
\rom{I}_\eps &= \frac{1}{2} \int\limits_{-\eps}^{\eps}
H(Z_{\eps}(w))\delta_{\eps}'(s)\dot{U}_\eps(w)\,\dd s
= 0 - \frac{1}{2}\int\limits_{-\eps}^\eps
\delta_\eps(s)\big(H(Z_\eps(w))\dot{U}_\eps(w)\dot{\big)}\,\dd s\\
&=-\frac{1}{2}\int\limits_{-\eps}^\eps
\delta_\eps(s)DH(Z_\eps(w))\dot{Z}_\eps(w)\,\dd s \\\nonumber
  &\quad + \frac{1}{2}\int\limits_{-\eps}^\eps
\delta_\eps(s)H(Z_\eps(w))\frac{\left(
\frac{1}{2}\dot{U}_\eps(w)\tilde{G}_\eps(w)-\big(H(Z_\eps)\delta_\eps U_\eps\dot{\big)}  (w) \right)s}{\sigma a^2 -
s^2H(Z_\eps(w))\delta_\eps(s)} \,\dd s + O(\eps)\,.
\end{align}
Now the integrals on the right-hand-side of \eqref{eq:I_eps} combine with
$\rom{II}_\eps$ and $\rom{V}_\eps$ to give

\begin{align}\label{eq:V_dot_neu}
  \dot{V}_{\eps} (\beta_\eps)&= \dot{V}_{\eps}^0 + \frac{1}{2}\rom{II}_\eps +
\rom{III}_\eps + \rom{IV}_\eps + \frac{1}{2}\rom{V}_\eps\,.
\end{align}
Now we insert equation \eqref{eq:geos} for $\dot{Z}_\eps$ into $\rom{II}_\eps$
and follow the same pattern as before. The only remarkable new point is the
occurrence of the regularisation-dependent term $\int_{-\eps}^\eps
\delta_\eps(s)^2 s\, \dd s$, whose prefactors cancel after a long and tedious
calculation, where we repeatedly use identities as e.g.\
$\int_{-\eps}^\eps\int_{-\eps}^s \delta_\eps(r)\,\dd r \,\dd s = \frac{1}{2}$.
For example we obtain for the term in \eqref{eq:V_dot_neu} related to
$\dot{Z}_\ep^0$
\begin{align*}\nonumber
 \Bigl| &\frac{1}{2}\int_{\alpha_\eps}^{\beta_\eps}
DH(Z_{\eps}(r))\delta_{\eps}(U_\eps(r))\dot{U}_\eps(r)\dot{Z}_\eps(r)\,\dd r - \frac{1}{2}DH(Z^0)\dot{Z}^0\Bigr| \\
 &=\frac{1}{2}\Bigl| \int_{-\ep}^\ep
\left(DH(Z_\ep(U_\ep^{-1}(s)))\delta_\ep(s)\dot{Z}_\ep(U_\ep^{-1}(s)) -
DH(Z^0)\delta_\ep(s)\dot{Z}^0\right)\,\dd s \Bigr| \\
&\qquad \leq \sup_{w\in U_\ep^{-1}([-\ep,\ep])}\frac{1}{2} \|\rho\|_{L^1}
\left( \left|DH(Z_\ep(w))\right| \left| \dot{Z}^0_\ep -
\dot{Z}^0 \right| + \left|\dot{Z}^0 \right| \left|DH(Z_\ep(w))-DH(Z^0)\right|\ \right) \to 0\,.
\end{align*}

\end{appendix}

\end{document}